\documentclass[11pt]{article}

\usepackage{amsmath}
\usepackage{amsthm}
\usepackage{amssymb}
\usepackage[usenames,dvipsnames]{xcolor}
\usepackage{graphicx}
\usepackage{longtable}
\usepackage[margin=1in]{geometry}
\usepackage[utf8]{inputenc}
\usepackage{scalefnt}
\usepackage[colorlinks=true, allcolors=blue]{hyperref}
\usepackage{mathtools}
\usepackage{scalerel}
\usepackage{makecell}

\usepackage{algorithm}
\usepackage{algorithmicx}
\usepackage[noend]{algpseudocode}

\usepackage{tikz}
\usetikzlibrary{shapes,matrix,calc,arrows,positioning,graphs}

\usepackage{float}

\usepackage{titlesec}
\titlespacing*{\section}
{0pt}{9pt plus 1pt minus 1pt}{7pt plus 1pt minus 1pt}
\titlespacing*{\subsection}
{0pt}{7pt plus 1pt minus 1pt}{5pt plus 1pt minus 1pt}

\newtheorem{theorem}{Theorem}[section]
\newtheorem{definition}[theorem]{Definition}
\newtheorem{proposition}[theorem]{Proposition}

\newtheorem{corollary}[theorem]{Corollary}
\newtheorem{lemma}[theorem]{Lemma}

\newtheorem{remark}[theorem]{Remark}

\newtheorem{assumption}{Assumption}
\newtheorem*{conjecture*}{Conjecture}

\newenvironment{replemma}[1]{%
  \begingroup
  \def\thetheorem{\ref{#1}}
  \begin{lemma}
}{%
  \end{lemma}
  \endgroup
}

\newenvironment{ccases}
  {\left\lbrace\begin{array}{>{\textstyle}c @{\quad} l}}
  {\end{array}\right.}
  
\newcommand{\AppendixBanner}{%
  \hbox to \textwidth{%
    \leaders\hbox{\rule[.55ex]{1em}{0.5pt}}\hfill
    \enspace\textbf{APPENDIX}\enspace
    \leaders\hbox{\rule[.55ex]{1em}{0.5pt}}\hfill
  }%
}

\makeatletter
\def\ps@appendix{%
  \let\@mkboth\markboth
  \def\@oddhead{\AppendixBanner}%
  \def\@evenhead{\AppendixBanner}%
  \def\@oddfoot{\hfil\thepage\hfil}%
  \def\@evenfoot{\hfil\thepage\hfil}%
}
\makeatother

\newcommand{\GE}{\;\ge\;}
\newcommand{\LE}{\;\le\;}
\newcommand{\EQ}{\;=\;}

\newcommand{\KT}{\,\mathrm{KT}}
\newcommand{\IT}{\,\mathrm{IT}}
\newcommand{\DT}{\,\mathrm{DT}}
\newcommand{\ST}{\,\mathrm{ST}}
\newcommand{\EC}{\,\mathrm{EC}}
\newcommand{\RT}{\,\mathrm{RT}}

\newcommand{\val}{\mathit{val}}
\newcommand{\flag}{\mathit{flag}}

\newcommand{\faked}{\mathsf{d}}

\newcommand{\SV}{Sharma and Vondr\'ak~\cite{SV14}}

\newcommand{\CKRal}{C{\u{a}}linescu et al.~\cite{CKR00}}
\newcommand{\KKSTYal}{Karger et al.~\cite{KKSTY04}}

\title{Improved Approximation Algorithms for Multiway Cut\\ by Large Mixtures of New and Old Rounding Schemes\thanks{Full version of paper to appear in the proceedings of STOC 2026. Companion code: \url{https://github.com/jbrakensiek/multiway-cut-verification}}}
\author{Joshua Brakensiek\thanks{University of California, Berkeley. Email: \texttt{josh.brakensiek@berkeley.edu}. Supported in part by a Simons Investigator award and NSF grants CCF-2211972 and DMS-2503280.}\and Neng Huang\thanks{University of Michigan. Email: \texttt{nengh@umich.edu}}\and Aaron Potechin\thanks{University of Chicago. Email: \texttt{potechin@uchicago.edu}}\and Uri Zwick\thanks{Blavatnik School of Computer Science, Tel Aviv University, Israel. Email: \texttt{zwick@tau.ac.il}. Work supported by ISF grant 2735/2025.}}
\date{}

\newcommand{\eps}{\varepsilon}
\newcommand{\E}{\mathop{\mathbb{E}}}

\newcommand{\RR}{\mathbb{R}}
\newcommand{\ee}{{\rm e}}
\newcommand{\dd}{{\rm d}}

\renewcommand{\Pr}{\mathbb{P}}

\newcommand{\be}{\boldsymbol{e}}

\newcommand{\bv}{\boldsymbol{v}}
\newcommand{\bu}{\boldsymbol{u}}

\newcommand{\bx}{{\boldsymbol{x}}}
\newcommand{\by}{{\boldsymbol{y}}}

\newcommand{\bR}{{\boldsymbol{R}}}
\newcommand{\bU}{{\boldsymbol{U}}}

\newcommand{\cR}{\mathcal{R}}

\newcommand{\argmin}{\operatorname{argmin}}
\newcommand{\argmax}{\operatorname{argmax}}

\newcommand{\true}{\mathtt{true}}
\newcommand{\false}{\mathtt{false}}

\newcommand{\leftsf}{\mathsf{L}}
\newcommand{\rightsf}{\mathsf{R}}

\newcommand{\ourAPX}{1.2787}  
\newcommand{\ourAPXd}{1.1489} 
\newcommand{\ourAPXe}{1.1837} 
\newcommand{\ourAPXf}{1.2149} 
\newcommand{\ourAPXg}{1.2399} 
\newcommand{\ourAPXh}{1.2499} 
\newcommand{\ourAPXi}{1.2549} 
\newcommand{\ourAPXj}{1.2599} 

\newcommand{\STprob}{0.138938}
\newcommand{\KTprob}{0.581747}
\newcommand{\ITprob}{0.120631}
\newcommand{\DTprob}{0.158685}
\newcommand{\numfun}{394}
\newcommand{\numfuna}{148} 

\newcommand{\MC}{Multiway Cut}

\begin{document}

\maketitle
\thispagestyle{empty} 

\begin{abstract}
\setlength{\parindent}{0pt}
\setlength{\parskip}{6pt plus 2pt}
\noindent
The input to the Multiway Cut problem is a weighted undirected graph, with nonnegative edge weights, and $k$ designated terminals. The goal is to partition the vertices of the graph into~$k$ parts, each containing exactly one of the terminals, such that the sum of weights of the edges connecting vertices in different parts of the partition is minimized. The problem is APX-hard for $k\ge3$. The currently best known approximation algorithm for the problem for arbitrary~$k$, obtained by Sharma and Vondr\'ak [STOC 2014] more than a decade ago, has an approximation ratio of 1.2965. We present an algorithm with an improved approximation ratio of 1.2787. Also, for small values of $k \ge 4$ we obtain the first improvements in 25 years over the currently best approximation ratios obtained by Karger, Klein, Stein, Thorup, and Young [STOC 1999]. (For $k=3$ an optimal approximation algorithm is known.)

Our main technical contributions are new insights on rounding the LP relaxation of C{\u{a}}linescu, Karloff, and Rabani [STOC 1998], whose integrality ratio matches Multiway Cut's approximability ratio, assuming the Unique Games Conjecture [Manokaran, Naor, Raghavendra, and Schwartz, STOC 2008]. First, we introduce a generalized form of a rounding scheme suggested by Kleinberg and Tardos [FOCS 1999] and use it to replace the Exponential Clocks rounding scheme used by Buchbinder, Naor, and Schwartz [STOC 2013] and by Sharma and Vondr\'ak. Second, while previous algorithms use a mixture of two, three, or four basic rounding schemes, each from a different family of rounding schemes, our algorithm uses a computationally-discovered mixture of hundreds of basic rounding schemes, each parametrized by a random variable with a distinct probability distribution, including in particular many different rounding schemes from the same family. We give a completely rigorous analysis of our improved algorithms using a combination of analytical techniques and interval arithmetic.

\end{abstract}

\clearpage
\setcounter{page}{1}

\section{Introduction}

The input to the (Minimum) Multiway Cut problem is a weighted undirected graph $G = (V,E,w)$, where $w:E\to\RR^+$, with $k$ designated vertices called \emph{terminals}. The objective is to find a partition of $V$ into $k$ parts, each containing exactly one terminal, so as to minimize the total weight of edges whose endpoints lie in different parts of the partition. For $k=2$ this is the classical min $s$-$t$ cut problem which admits efficient polynomial-time algorithms. For $k\ge 3$, however, the problem is APX-hard~\cite{DJPSY94}. This motivated a long series of papers \cite{DJPSY94,CKR00,KKSTY04,CCT06,SV14,BNS18,BSW19,BSW21} that design approximation algorithms for the problem, with gradually improving approximation ratios. The currently best approximation algorithm for arbitrary $k$, due to Sharma and Vondrák~\cite{SV14}, achieves a ratio of $1.2965$. For the full history of the problem, and a summary of all previous results, see Table~\ref{T-approx} and Section~\ref{sub:previous}.

For $k=3$, Karger, Klein, Stein, Thorup, and Young \cite{KKSTY04} and independently Cheung, Cunningham, and Tang \cite{CT99,CCT06}, combined with a result of Manokaran, Naor, Raghavendra, and Schwartz \cite{MNRS08}, showed that the best approximation ratio for the problem, assuming the Unique Games Conjecture (UGC), is $\frac{12}{11}=1.0909\ldots\;$. Karger et al.~\cite{KKSTY04} also give improved approximation algorithms for small values of~$k$. These results are summarized in Table~\ref{T-finite}.

Our main result is an improved approximation algorithm for the \MC\ problem, for arbitrary~$k$, whose approximation ratio is $\ourAPX$. This is the first improved approximation ratio for the problem for over a decade.

\begin{theorem}\label{thm:main}
\MC\ with an arbitrary number of terminals has a polynomial-time $\ourAPX$-approximation algorithm.
\end{theorem}

We also obtain improved approximation ratios for every value of $k\ge 4$. For the exact approximation ratios, see Table~\ref{T-finite}. These are the first improvements of these approximation ratios for over 25 years. In particular, we have:

\begin{theorem}\label{thm:finite}
(a) \MC\ with $4$ terminals has a polynomial-time $\ourAPXd$-approximation algorithm. (b) \MC\ with $5$ terminals has a polynomial-time $\ourAPXe$-approximation algorithm. (c) \MC\ with $6$ terminals has a polynomial-time $\ourAPXf$-approximation algorithm. 
\end{theorem}

\subsection{Previous Work}\label{sub:previous}

Dahlhaus et al.~\cite{DJPSY92,DJPSY94} were the first to consider the \MC\ problem from the computational complexity perspective. They proved that the problem is APX-hard for $k\ge 3$ and gave a simple combinatorial algorithm with an approximation ratio of $2-\frac{2}{k}$. They also showed that in planar graphs, the problem can be solved in polynomial time for every fixed $k$, but is NP-hard when~$k$ is part of the input. 

Cunningham~\cite{C91}, Chopra and Rao~\cite{CR91}, and Bertsimas, Teo, and Vorha~\cite{BTV99} studied the facets of the multiway cut polyhedron and suggested some simple linear programming relaxations of the problem. These did not lead to improved approximation ratios.

The next major step was made by C{\u{a}}linescu, Karloff, and Rabani~\cite{CKR00}. They introduced a stronger LP relaxation of the problem, equivalent to embedding the vertices of the graph into the $k$-simplex $\Delta_k=\{(u_1,u_2,\ldots,u_k) \mid u_1,u_2,\ldots,u_k\ge 0\,,\,\sum_{i=1}^k u_i=1\}$, \footnote{According to our definition, the $k$-simplex has $k$ vertices and is $(k-1)$-dimensional. Some authors refer to this as the $(k-1)$-simplex. Our notation seems more convenient for our purposes.} and used a simple \emph{Single Threshold (ST)} rounding scheme to obtain an approximation algorithm with an approximation ratio of $1.5-\frac{1}{k}$. All subsequent approximation algorithms for the \MC\ problem use the relaxation of C{\u{a}}linescu et al.~\cite{CKR00}. For completeness, the CKR LP relaxation is described in Appendix~\ref{A:LP-relax}.

Karger et al.~\cite{KKSTY99,KKSTY04} obtained an improved approximation ratio of $1.3438$. Their algorithm uses the Single Threshold (ST) rounding scheme of C{\u{a}}linescu et al.~\cite{CKR00} together with another rounding scheme which is now known as \emph{Independent Thresholds (IT)}. 

\begin{table}[t]
\begin{center}
\begin{tabular}{c@{\hspace{3em}}ccc}
Ratio & Paper & Method & Analysis Type\\[3pt]
\hline
\noalign{\vskip 3pt}\noalign{\vskip 3pt}
2 & \cite{DJPSY94} & Combinatorial & Analytic\\
1.5 & \cite{CKR00} & ST & Analytic\\
1.3438 & \cite{KKSTY04} & ST + IT & Analytic\\
1.3239 & \cite{BNS18} & ST + EC/KT & Analytic\\
1.3022 & \cite{SV14} & ST + EC/KT + DT & Analytic\\
1.2970 & \cite{BSW21} & Simplex Trans. + ST + EC/KT & Analytic\\
1.2965 & \cite{SV14} & ST + EC/KT + IT + DT & Computational\\[3pt]
\hline
\noalign{\vskip 3pt}\noalign{\vskip 3pt}
\ourAPX & Here &  ST + GKT${}^*$ + IT${}^*$ + DT${}^*$ & Computational\\
\end{tabular}
\end{center}
\caption{A summary of \MC\ approximation algorithms for an arbitrary number of terminals~$k$. The ``Method'' column describes the rounding techniques used by the algorithms according to the following legend:  Single Threshold (ST), Independent Thresholds (IT), Exponential Clocks or Kleinberg-Tardos (EC/KT), Descending Thresholds (DT), and Generalized Kleinberg-Tardos (GKT). A star, as in IT${}^*$, indicates that several independent copies of the rounding schemes of this families are used, each parametrized by a different random variable.
}\label{T-approx}
\end{table}

As mentioned, Cheung, Cunningham, and Tang \cite{CT99,CCT06} and independently Karger et al.~\cite{KKSTY99,KKSTY04}, obtained a $\frac{12}{11}$-approximation ratio for the case $k=3$. They also showed that this matches the integrality ratio of the relaxation of C{\u{a}}linescu et al.~\cite{CKR00}.

Karger et al.~\cite{KKSTY99,KKSTY04} also obtained improved approximation ratios for small values of~$k$. For $k=4,5$, their best algorithms are obtained by discretizing the simplex and using very large linear programs to find nearly-optimal distributions of \emph{side-parallel cuts (sparcs)}. For $k\ge 6$, they use their algorithm for general~$k$ with slightly tuned parameters. The exact ratios obtained can be found in Table~\ref{T-finite}.

The next improvement, for general~$k$, was obtained by Buchbinder, Naor, and Schwartz \cite{BNS18}. They obtained a simple approximation algorithm with an approximation ratio of $\frac43 - \frac{4}{9k-6}<1.3334$ and a slightly more complicated algorithm with an approximation ratio of $1.3239 - \frac{1}{24k}$. Their algorithms use a mixture of the Single Threshold (ST) rounding scheme with a non-uniform random variable, and a technique called \emph{Exponential Clocks (EC)} that was previously introduced by Ge, He, Ye, and Zhang~\cite{GHYZ11}. They also showed that the Exponential-Clocks (EC) rounding scheme could be replaced by a rounding scheme introduced by Kleinberg and Tardos~\cite{KlTa02}.

\begin{table}[t]
\begin{center}
\vspace*{10pt}
\begin{tabular}{cccccc}
$k$ & \makecell[c]{SPARCs\\ \cite{KKSTY04}} & \makecell[c]{ST+IT\\ \cite{KKSTY04}} & \makecell[c]{EC+ST\\ \cite{BNS18}} & \makecell[c]{New\\ Algorithms} & \makecell[c]{Lower Bound\\ \cite{AMM17}}\\[10pt]\hline
\noalign{\vskip 3pt}
4 & 1.1539 & 1.189 & 1.2000 & \ourAPXd & 1.125\\
5 & 1.2161 & 1.223 & 1.2308 & \ourAPXe & 1.142\\
6 & 1.2714 & 1.244 & 1.2500 & \ourAPXf & 1.1538\\
7 & 1.3200 & 1.258 & 1.2632 & \ourAPXg & 1.1612\\
8 & 1.3322 & 1.269 & 1.2728 & \ourAPXh & 1.1666\\
9 & -- & 1.277 & 1.2800 & \ourAPXi & 1.1707\\
10 & -- & 1.284 & 1.2858 & \ourAPXj & 1.1739\\
\hline
\noalign{\vskip 3pt}
any &  &  &  & \ourAPX & 1.20016\\
 &  &  &  &  & \cite{BCKM20}\\
\end{tabular}
\vspace{-4pt}
\end{center}
\caption{Upper and lower bounds on the approximation ratios of the \MC\ problem with $k=4,5,\dots,10$ terminals. (For $k=3$ a tight bound of $\frac{12}{11} = 1.0909...$ is known.) The lower bounds are on the integrality ratio of the LP relaxation of \MC, which translate under the Unique Games Conjecture to lower bounds on the approximability ratio of the problem. The last line restates the approximation ratio that we get for an arbitrary number of terminals and also gives the best known lower bound for this case. The approximation ratio that we obtain for general~$k$ improves over all previously known results for $k\ge 10$.}\label{T-finite}
\end{table}

Sharma and Vondr\'ak~\cite{SV14} built on the results of Buchbinder et al.~\cite{BNS18} and obtained three improved approximation algorithms for the \MC\ problem. Their first algorithm, with a ratio of $\frac{3+\sqrt{5}}{4}\simeq 1.309017$, was obtained using a combination of the Single Threshold (ST) and Exponential-Clocks (EC) rounding techniques, as used by Buchbinder et al.~\cite{BNS18}. They also showed that under some mild assumptions, this is the best ratio that can be obtained by combining these two techniques. Their second approximation algorithm, with a ratio of $\frac{10+4\sqrt{3}}{13}\simeq 1.30217$, was obtained by combining the two original rounding schemes with a new rounding technique which they call \emph{Descending Thresholds (DT)}. Their third and final algorithm, with an approximation ratio of $1.2965$, is obtained by adding the Independent Thresholds (IT) rounding technique of Karger et al.~\cite{KKSTY04} to the mix. Their analysis of the third algorithm is computer assisted.

Buchbinder, Schwartz, and Weizman \cite{BSW21} obtained a relatively simple algorithm with an approximation ratio of $\frac{297}{229} \simeq 1.29694$, almost matching the best approximation ratio of Sharma and Vondr\'ak~\cite{SV14}. This is currently the best approximation ratio that can be verified analytically without the use of a computer. Their algorithm uses an interesting technique called \emph{simplex transformations}. As a result, unlike all algorithms discussed above, their algorithm uses cuts that are not parallel to one of the faces of the simplex, i.e., cuts that depend on more than one coordinate. Buchbinder, Schwartz, and Weizman \cite{BSW19} used a global linear transformation to obtain an approximation ratio of $\frac{11}{8}=1.375$ using just the Single Threshold (ST) rounding technique.

Manokaran, Naor, Raghavendra, and Schwartz \cite{MNRS08} showed that, under the Unique Games Conjecture (UGC) of Khot \cite{khot02}, the integrality ratio of the LP relaxation of C{\u{a}}linescu et al.~\cite{CKR00} is the best approximation ratio that can be obtained for the \MC\ problem in polynomial time. Thus, efforts for obtaining improved approximation algorithms for the problem should concentrate on finding better ways of rounding solutions of this LP relaxation. In particular, no improved results can be obtained by trying to use semidefinite programming (SDP) relaxations.

Freund and Karloff~\cite{FK00} showed that the integrality ratio of the CKR relaxation is at least $8/(7+\frac{1}{k-1})$, for any $k\ge 3$. Angelidakis, Makarychev, and Manurangsi~\cite{AMM17} improved the lower bounds to $6/(5+\frac{1}{k-1})$ by introducing an interesting new framework of \emph{non-opposite cuts}. B{\'e}rczi, Chandrasekaran, Kir{\'a}ly, and Madan~\cite{BCKM20} used this framework to obtain a lower bound of $1.20016$ for general~$k$.

\subsection{Our Contributions and Techniques}

Our improved approximation algorithms for \MC\ are obtained by designing improved rounding schemes for the LP relaxation of C{\u{a}}linescu et al.~\cite{CKR00}. Previous algorithms used a mixture of up to four basic rounding schemes, each taken from a different family of rounding schemes. As mentioned the families of rounding schemes used were: Single Threshold (ST), Independent Thresholds (IT), Exponential clocks (EC) or Kleinberg-Tardos (KT), and Descending Thresholds (DT). We improve on all previous results using the following steps:
\begin{enumerate}
    \item We introduce a generalized version of the Kleinberg-Tardos (KT) rounding scheme.
    \item We use many basic rounding schemes from each family of rounding schemes, each parametrized by a random variable with a different distribution.
    \item We devise an algorithm for finding good mixtures of a very large number of basic rounding schemes by approximating the solution of an infinite $0$-sum game played by a cut player and an edge player.
    \item We use rigorous computational techniques, in particular \emph{Interval Arithmetic}, to obtain rigorous computer-assisted proofs of the approximation ratios obtained.
\end{enumerate}

We elaborate on each of these steps in the following sections.

\subsubsection{Generalized Kleinberg-Tardos}

A (multiway) \emph{cut} of the $k$-simplex $\Delta_k=\{(u_1,u_2,\ldots,u_k) \mid u_1,u_2,\ldots,u_k\ge 0\,,\,\sum_{i=1}^k u_i=1\}$ is a function $c: \Delta_k \to [k]$ that maps each point of the simplex to one of the~$k$ terminals, such that $c(\be_i)=i$, for $i\in[k]$, i.e., the $i$-th vertex of the simplex is mapped to terminal~$i$. A \emph{rounding scheme} is a probability distribution over cuts.

Kleinberg and Tardos~\cite{KlTa02} introduced the following rounding scheme, which we call KT, and used it to obtain a $2$-approximation algorithm for the Uniform Metric Labeling problem. Initially all the simplex points are unassigned. The rounding scheme proceeds in rounds. In each round, a uniformly random terminal $i\in[k]$ and a uniformly random \emph{threshold} $t\in[0,1]$ are chosen. Each unassigned simplex point $\bu=(u_1,u_2,\ldots,u_k)$ with $u_i\ge t$ is assigned to terminal~$i$, i.e., $c(\bu)=i$. This proceeds until all simplex points are assigned. This rounding scheme, unlike all other schemes we consider, samples the terminals \emph{with repetitions}. It is not difficult to check that the expected number of rounds needed to `color' the whole simplex is polynomial in~$k$.

Buchbinder, Naor, and Schwartz \cite{BNS18} showed that the rounding technique of Kleinberg and Tardos~\cite{KlTa02} can be used to replace the Exponential Clocks (EC) rounding technique which they and \SV\ use in their approximation algorithms, as both rounding techniques, although very different from each other, have identical \emph{cut density} functions, as we define below. (For a more formal definition, see Section~\ref{sec:rounding}.) 

A Generalized Kleinberg-Tardos (GKT) rounding scheme is a rounding scheme in which the random thresholds are chosen \emph{nonuniformly}. More precisely, we let KT$(f)$, where $f:[0,1]\to\RR^+$ is a density function, be the Generalized Kleinberg-Tardos rounding scheme in which the random thresholds are chosen independently according to a random variable with density function~$f$. To ensure termination we require that $F(\frac1k)>0$, where $F(x)=\int_0^x f(t)\,dt$ is the cumulative probability function corresponding to~$f$. It is interesting to note that while nonuniform thresholds were used in conjunction with many of the other random techniques, this was not tried with the Kleinberg-Tardos rounding scheme. Our results show that choosing the thresholds nonuniformly greatly enhances the performance of such rounding schemes.

If $R$ is a rounding scheme and $\bu,\bv\in\Delta_k$ are two simplex points, we let $P^{R}_k(\bu,\bv)=\Pr[c(\bu)\ne c(\bv)]$ be the probability that the edge $(\bu,\bv)$ is cut by the rounding scheme $R$. If~$R$ is \emph{symmetric}, i.e., it treats all terminals equally\footnote{More specifically, $P^{R}_k(\bu,\bv)$ is invariant under permutation of terminals. All rounding schemes considered in this paper are symmetric. In fact, one can show that any rounding scheme can be replaced with a symmetric one without increasing the worst-case cut density.}, we can restrict our attention to \emph{$(1,2)$-aligned edges}, i.e., edges of the form $\bu=(u_1,u_2,\ldots,u_k)$ and $\bv=\bu-\eps(\be_1-\be_2)=(u_1-\eps,u_2+\eps,u_3,\ldots,u_k)$. We define the \emph{density} of the rounding scheme~$R$ at~$\bu$ as follows: $d^R_k(\bu)=\lim_{\eps\to 0} \frac{P^{R}_k(\bu,\bu-\eps(\be_1-\be_2))}{\eps}$. It is known that $\sup_{\bu\in\Delta_k}d^R_k(\bu)$ is the approximation ratio achieved by using the rounding scheme~$R$.

One of the features that make the Kleinberg-Tardos (KT) and the Exponential Clocks (EC) rounding schemes attractive is that they have very simple density functions $d^{\KT}_k(\bu)=d^{\EC}_k(\bu)=2-u_1-u_2$. Used on their own, each one of these rounding schemes only gives an approximation ratio of~$2$ for the \MC\ problem. However, they are useful in combination with other rounding techniques.

As a first step towards using Generalized Kleinberg-Tardos schemes KT$(f)$ to obtain improved approximation algorithms, we obtain the following relatively simple formula for their densities:

\begin{proposition}[See Corollary~\ref{C-KT}]\label{prop:KT-intro}
Let $f:[0,1]\to\RR^+$ be a probability density function and let $F : [0,1] \to [0,1]$ be its corresponding cumulative density function. Then,
\[d^{\KT(f)}_{k}(u_1,u_2,\ldots,u_k) \EQ
    \frac{ f(u_1) }{ \sum_{i=1}^k F(u_i) }\left(1 - \frac{F(u_1)}{\sum_{i=1}^k F(u_i)} \right) \,\,+\,\, \frac{ f(u_2) }{ \sum_{i=1}^k F(u_i) }\left(1 - \frac{F(u_2)}{\sum_{i=1}^k F(u_i)} \right)\;.\]
\end{proposition}

Note that if $f(x)=1$ and $F(x)=x$, the density simplifies to $2-u_1-u_2$. 

By varying the distribution $f$ over $[0,1]$, we get access to a large family of cut density functions that were not available to prior approaches to \MC. However, to take advantage of this newfound freedom, we also need a novel approach to discovering \emph{mixtures} of rounding schemes with a low approximation ratio, as discussed below.

A possible alternative approach is to consider generalized versions of the Exponential Clocks (EC) rounding technique in which a distribution other than the exponential distribution is used for the `clocks'. We discuss this option briefly in Appendix~\ref{app:GEC}. The main difficulty with this approach is that the corresponding density functions seem hard to work with.

\subsubsection{Using Large Mixtures of Rounding Schemes}

If $R_1,R_2,\ldots,R_n$ are rounding schemes, and $(p_1,p_2,\ldots,p_n)\in\Delta_n$, the \emph{mixture} $\cR=\sum_{i=1}^n p_iR_i$ is the rounding scheme that applies $R_i$ with probability $p_i$, for $i\in[n]$. We refer to the $R_i$'s as the \emph{basic} rounding schemes that constitute~$\cR$. By the linearity of expectations, we have $d^{\cR}_k(\bu) = \sum_{i=1}^n p_i\, d^{R_i}_k(\bu)$. Mixtures allow an algorithm designer to balance out several rounding schemes whose densities are maximized at different points of the $k$-simplex.

As mentioned above, most approximation algorithms for \MC\ use a mixture of several basic rounding schemes, each taken from a different family of rounding schemes. \SV, for example, use a mixture of four basic rounding schemes, ST$(f)$, EC/KT, IT$(g)$, and DT$(g)$, where~$f$ is a rather complicated, computationally discovered, density function and~$g$ is the density function of a uniformly random variable on $[0,\frac{6}{11}]$. (See the legend given in Table~\ref{T-approx} and Section~\ref{sec:rounding}.)

We note that with the exception of Single Threshold (ST), the other families of rounding schemes, including the newly introduced family KT$(f)$ of generalized Kleinberg-Tardos rounding schemes, are \emph{nonlinear}, in the sense that $\alpha_1 \KT(f_1) + \alpha_2 \KT(f_2) \ne \KT(\alpha_1 f_1 + \alpha_2 f_2)$. In fact, $\alpha_1 \KT(f_1) + \alpha_2 \KT(f_2)$ does not seem to be equivalent to KT$(g)$ for any density function~$g$. The same observation holds also for Independent Thresholds (IT) and Descending Thresholds (DT).

As a consequence, it may be beneficial to use a mixture that contains many rounding schemes from each family of rounding schemes, each with its distinct density function. Thus, for example, we may use a mixture that contains KT$(f_1)$, KT$(f_2)$,\ldots, KT$(f_n)$, for an arbitrarily large value of~$n$, and similarly IT$(g_1)$, IT$(g_2)$,\ldots, IT$(g_m)$, etc. Our results show that using such large mixtures is indeed beneficial and may lead to substantially improved approximation ratios.

This observation substantially broadens the range of possibilities. Exploring this huge space manually seems practically impossible. We thus need an automated technique for discovering good mixtures of rounding schemes, mixtures that could possibly contain hundreds of basic rounding schemes.

\subsubsection{Computational Discovery of Rounding Schemes}

Computational techniques were used in several of the previous papers on the \MC\ problem. Both Cheung et al.~\cite{CT99,CCT06} and Karger et al.~\cite{KKSTY04} solved large linear programs, obtained by discretizing the $3$-simplex, to discover similar, but not identical, optimal algorithms for $k=3$ and to obtain matching lower bounds on the integrality ratio of the relaxation. Once the optimal algorithms were discovered, they were able to analyze them without the use of a computer.

Karger et al.~\cite{KKSTY04} also used two different computational approaches to obtain their algorithms for $k\ge 4$. For $k=4,5$ they again solved large linear programs, somewhat different from the ones used for $k=3$, to obtain their best algorithms for these values. For $k\ge 6$ they obtained their best algorithms using a mixture of ST$(f)$ and IT$(g)$, where~$f$ is the density function of a uniform random variable on $[b,1]$, and~$g$ is the density of a uniform random variable on $[0,b]$, where $b\in [0,1]$ is a parameter that they tuned computationally. (For each value of~$k$ they use a different value of~$b$.) 

\SV\ used a mixture of ST$(f)$, EC/KT, IT$(g)$ and DT$(g)$, where $f$ is a complicated density function discovered computationally and~$g$ is the density function of a uniformly random variable on $[0,\frac{6}{11}]$. Because of the linearity of the density function of ST schemes, they were able to solve a huge linear program that optimizes a very fine discrete version of~$f$, and the probabilities with which each one of their four basic rounding schemes should be used.

We use much more extensive computational techniques to discover rounding schemes that are mixtures of hundreds of basic rounding schemes, each with its own density function. This involves solving many nonlinear, and non-convex, optimization problems. We do that using standard numerical optimization techniques. (More details are given in Section~\ref{sec:minimax}.) These techniques are, of course, not guaranteed to find globally optimal solutions. However, once we get a suggested rounding scheme, we can rigorously obtain an upper bound on its approximation ratio.

We next give a high-level description of our computational techniques. Let $R_1,R_2,\ldots,\allowbreak R_n$ be a finite collection of rounding schemes and let $\bu_1,\bu_2,\ldots,\bu_m$ be a finite collection of simplex points. We can use a linear program to find a mixture $\cR=\sum_{i=1}^n p_iR_i$ of the rounding schemes that minimizes $\max_{j\in [m]} d^\cR_k(\bu_j)$. The dual linear program finds a distribution $(q_1,q_2,\ldots,q_m)$ over the available simplex points that maximizes $\min_{i\in [n]} \sum_{j=1}^m q_j\, d^{R_i}_k(\bu_j)$. This can thus be viewed as a \emph{$0$-sum game}, with payoff matrix $(d^{R_i}_k(\bu_j))$, played between the algorithm designer, who chooses rounding schemes, and the adversary, who chooses simplex points. Viewed from this perspective, it is not surprising that in most cases, the optimal strategies of the two players are \emph{mixed}, i.e., a probability distribution over either the rounding schemes or the simplex points. Such linear programs were used by many of the previous papers, including \cite{CT99,CCT06,KKSTY04,SV14}.

We go one step further. Suppose that $(p_1,p_2,\ldots,p_n)$ and $(q_1,q_2,\ldots,q_m)$ are the optimal mixed strategies of the players in the above game. We try to find a new rounding scheme $R'$ for which $\sum_{j=1}^m q_j\, d^{R'}_k(\bu_j)$ is smaller than the value of the game and add $R'$ to the collection of basic rounding schemes available to the algorithm designer. Similarly, we also try to find a new simplex point $\bu'$ for which $d^{\cR}_k(\bu')=\sum_{i=1}^n p_i\,d^{R_i}_k(\bu')$ is larger than the value of the game and add it to the collection of simplex points. After each such step we resolve the $0$-sum game and obtain the newly optimal distributions. We let this process run for as long as we can, or until no significantly improved results are found. For more details, see Section~\ref{sec:minimax}.

We note that a similar computational technique was used by Brakensiek, Huang, Potechin, and Zwick \cite{BHPZ23} to discover improved approximation algorithms for MAX DI-CUT and other MAX CSP problems.

\subsubsection{Verification Using Interval Arithmetic}\label{sub-verify}

An interval arithmetic system provides functions for computing the basic arithmetic operations, and other basic functions such as $\exp(x)$, on \emph{intervals} rather than on numbers. The functions return an interval that is guaranteed to contain the correct interval. For example, if $f_{int}$ is an interval arithmetic implementation of a function~$f$, and $I=[a,b]$ is an interval, then $f_{int}(I)\supseteq f(I)=\{ f(x) \mid x\in I\}$. Note that $f_{int}$ is not required to return $f(I)$ exactly, as it may be impossible to represent the true endpoints of the interval using the machine precision, or it may be too time consuming to compute them exactly. To be useful, however, $f_{int}(I)$ should not be much larger than~$f(I)$. By combining the interval implementations of the basic functions, interval implementation of much more complicated functions can be obtained. It is important to note that the interval implementation of the basic functions should take into account all numerical errors, including floating point rounding errors, that may occur during the computation of the final interval. We make use of a modern interval arithmetic library called Arb \cite{johansson2017arb,johansson2019computing}.

To illustrate the usefulness of interval arithmetic, consider the following example. Suppose that $f(x_1,x_2,\ldots,x_k)$ is a continuous, and possibly differentiable, function defined on $[0,1]^k$. Suppose that $m = \min_{x_1,\ldots,x_k\in[0,1]} f(x_1,x_2,\ldots,x_k)$ and that $m>0$. Our goal is to prove that $m>0$. A traditional approach is to compute the partial derivatives of~$f$, assuming that $f$ is differentiable, look for critical points and show that $f$ is positive at all critical points. This is actually not enough as we also need to look for critical points on the boundary of the region. In many cases this approach is not feasible as the partial derivatives may be very complicated and there may be an enormous number of critical points. If we can show that $f(x_1,x_2,\ldots,x_k)$ is $L$-Lipschitz, i.e., $|f(\bx)-f(\by)|\le L\,\|\bx-\by\|_1$, for every $\bx,\by\in[0,1]^k$, then in order to obtain the worst-case cut density within any given additive error $\epsilon > 0$, it is enough to evaluate $f$ on all points of the form $(\frac{i_1}{N},\frac{i_2}{N},\ldots,\frac{i_k}{N})$, where $i_1,i_2,\ldots,i_k\in \{0,1,\ldots,N\}$, for a sufficiently large integer~$N = N(\epsilon, L)$. However, this requires a rigorous bound on the Lipschitz constant~$L$, and rigorous bounds on the numerical errors that could occur during all these function evaluations, again not a pleasant task. Using interval arithmetic, on the other hand, we can adaptively partition $[0,1]^k$ into a collection of boxes $I_1\times I_2\times\cdots\times I_k$ and compute $f_{int}(I_1, I_2\,\dots,I_k)$ for each one of these boxes. If all returned intervals are positive, we get a rigorous proof that $m>0$. In particular, this approach neither requires that $f$ is differentiable nor that we have an explicit bound on the Lipschitz constant~$L$. Furthermore, any numerical errors are automatically accounted for by the interval arithmetic system.

Interval arithmetic was used before to obtain rigorous proofs of approximation ratios. (See, e.g.,~\cite{zwick02, Sjogren09, AuBeGe16, BhangaleKKST18, BhKh20, BHPZ23, BHZ24}.) We note that Sharma and Vondr{\'a}k \cite{SV14} used computational techniques to obtain a bound of $1.2965$ on the approximation ratio of their algorithm, but did not explicitly bound the numerical errors that may have occurred during the computation of this constant. (It is extremely unlikely, however, that these unavoidable numerical errors affect the fourth decimal digit of their bound.) The same holds for some of the bounds obtained by Karger et al.~\cite{KKSTY04} for small values of~$k$.

Since our rounding schemes use mixtures of hundreds of basic rounding schemes, analyzing their performance by hand is infeasible. We therefore resort to interval arithmetic to rigorously bound their performance ratios. For small values of~$k$ we could in principle try to use the general approach outlined above, with the $k$-dimensional box $[0,1]^k$ replaced by the $k$-simplex $\Delta_k$. However, even for relatively small values of~$k$ this becomes infeasible and we need to use additional ideas. The basic approach fails completely when $k$ is arbitrary, as then we are essentially trying to minimize a function defined on $\cup_{k\ge 3}\Delta_k$.

To overcome these obstacles, our verification works with \emph{prefixes} of simplex points instead of simplex points themselves. A prefix of a simplex point is a tuple $\bu = (u_1, \ldots, u_\ell)$ where $0 \leq u_1, \ldots, u_\ell \leq 1$ and $\sum_{i = 1}^\ell u_i \leq 1$. Such a prefix can be extended to a point in the $k$-simplex, for any $k>\ell$, by appending additional coordinates $u_{\ell+1},\ldots,u_k$, making sure that $\sum_{i=1}^k u_i=1$. (If $\sum_{i = 1}^\ell u_i = 1$, then $\bu$ is already a simplex point.) The density of a prefix~$\bu$ is defined to be the maximum density of any simplex point that can be obtained by extending~$\bu$. The main purpose of working with prefixes is that we may use one evaluation, namely, evaluation of the density of the prefix, to cover all simplex points having that prefix. Bounding densities of prefixes is a nontrivial task. Karger et al.~\cite{KKSTY04} and Sharma and Vondr{\'a}k \cite{SV14} used a similar approach in their analytical and computational analyses, but only with prefixes of length $\ell=2$. To get our improved approximation ratios we need to consider longer prefixes.

The main technical difficulty now is to obtain a good upper bound on the density of prefixes of any length. It turns out that by imposing very mild conditions on our rounding schemes, we can derive a \emph{tight} upper bound on the density of not all, but a sizable chunk of simplex points having the prefix. More specifically, if we assume that in our mixture $\mathcal{R}$, there exists some $\alpha > 0$ such that for any $\IT(f) \in \mathcal{R}$, $f$ is increasing on $[0, \alpha]$, and for any $\KT(g) \in \mathcal{R}$, $g$ is constant on $[0, \alpha]$, then we may define a function $\faked^R$ for every $R \in \mathcal{R}$ such that for any prefix $\bu = (u_1, \ldots, u_\ell)$ and $k \geq \ell$, 

\begin{equation}\label{eq:smeared_upper_bound}
\max_{\substack{0\leq u_{\ell + 1}, \ldots, u_{k} \leq \alpha \\ \sum_{i = 1}^ku_i = 1}}\,\,d^{R}_k(u_1,u_2,\ldots,u_\ell, u_{\ell + 1}, \ldots, u_k) \LE \faked^R(\bu).
\end{equation}
Furthermore, equality is achieved most of the time (See Lemmas~\ref{lem:upper_bound_general_k} and~\ref{lem:upper_bound_small_k} for more details). The evaluation for the expression on the right hand side is straightforward for ST, GKT, and DT schemes. For IT schemes it is more complicated, for we need to derive an expression involving Kummer's confluent hypergeometric function (see Lemma~\ref{L-IT-asymptotic}). 

This leads naturally to our verification algorithm, for which we now give a high-level description. Let us say we are attempting to verify that the density achieved by $\mathcal{R}$ (which satisfies the above conditions regarding IT and GKT schemes) is always at most $d_0$ for some $d_0 > 0$.
Using interval arithmetic, we work with sets of prefixes $\bU = U_1 \times \cdots \times U_\ell$, where $U_1, \ldots, U_\ell$ are intervals within $[0, 1]$. We implement in interval arithmetic the evaluation of $\E_{R \in \mathcal{R}}[\faked^R(\bU)]$, which gives us an interval containing evaluations of $\E_{R \in \mathcal{R}}[\faked^R(\bu)]$ for prefixes $\bu \in \bU$. If we find some $\bu \in \bU$ for this evaluation is greater than $d_0$, then we get a prefix which could potentially be extended to a point on which the density of $\mathcal{R}$ is greater than $d_0$ (in fact, it can in most cases). In this case we halt the verification without success. Otherwise, by \eqref{eq:smeared_upper_bound} we know that for every simplex point with prefix in~$\bU$ whose remaining coordinates are all at most $\alpha$, its cut density is at most $d_0$. So we consider this level ``passed'', and add a new interval $U_{\ell + 1} = [\alpha, 1]$ and recursively apply this algorithm. Since the coordinates of any simplex point sum to 1, we can add at most $1/\alpha$ new coordinates, which guarantees that the verification will terminate in a finite amount of time. This strategy may also be adapted for verification of a finite fixed $k$. See Section~\ref{sec:verification} for more details.

\subsection{Related Work}

The \MC\ problem is a special case of the \emph{0-Extension} problem, introduced by Karzanov \cite{Karzanov98}, and the \emph{Metric Labeling} problem introduced by Kleibert and Tardos~\cite{KlTa02}. In the 0-Extension problem, in addition to the weighted undirected graph $G=(V,E,w)$, where $w:E\to\RR^+$, the input also includes a \emph{metric}~$d$ defined on the set of terminals~$T$. To goal is to find an assignment $\ell:V\to T$ that minimizes the separation cost $\sum_{\{u,v\}\in E} w(u,v)\, d(\ell(u),\ell(v))$. If $d$ is the uniform metric, i.e., $d(t,t')=1$, if $t\ne t'$, and $d(t,t)=0$, this is exactly the \MC\ problem. 

In the Metric Labeling problem, the input also includes an \emph{assignment cost} function  $c:V\times T\to \RR^+$ and the goal is to find an assignment $\ell:V\to T$ that minimizes $\sum_{u\in V} c(u,\ell(u))+\sum_{\{u,v\}\in E} w(u,v)\, d(\ell(u),\ell(v))$, the sum of the assignment and separation costs. A special case of the Metric Labeling problem that still captures the \MC\ problem is the \emph{Uniform Metric Labeling} problem in which the metric~$d$ is the uniform metric.

The best approximation ratio known for the $0$-Extension problem is $O(\log k/\log\log k)$ \cite{CKR05,FHRT03}, where $k=|T|$ is the number of terminals. The best approximation ratio known for the Metric Labeling problem is $O(\log k)$ \cite{KlTa02}. The best approximation ratio known for the Uniform Metric Labeling problem is~$2$ \cite{KlTa02}.

Manokaran et al.~\cite{MNRS08} showed that, under the Unique Games Conjecture (UGC), the best approximation ratios that can be obtained for the $0$-Extension and the Metric Labeling problems in polynomial time are equal to the intergality ratios of the \emph{earth mover} LP relaxations of these problems introduced by Chekuri, Khanna, Naor, and Zosin~\cite{CKNZ04}.

A complementary problem to \MC\ is \emph{Multiway Uncut}, where one seeks to maximize the number (or total weight) of edges which are \emph{not} cut by an assignment to the terminals. The exact optimization problem is equivalent to \MC, but the quantitative guarantees of approximation algorithms are rather different. Langberg, Rabani, and Swamy~\cite{LRS06} studied Multiway Uncut and obtained an approximation ratio of $0.8535$, with a nearly matching LP integrality gap of $\frac{6}{7}\approx 0.8571$.

Ene, Vondr{\'a}k, and Wu~\cite{EVW13} studied a number of generalizations of \MC. The authors divided these generalizations into two categories (1) submodular optimization problems generalization \MC\ (Submodular Multiway Parition) and (2) Min-CSP problems generalizing \MC\ (e.g., \emph{Hypergraph} \MC). In both cases, they showed that in general the optimal approximation ratio is $2-\frac{2}{k}$ assuming either $RP \neq NP$ or the Unique Games Conjecture. They also obtained better approximation ratios in some special cases.

Very recently, B{\'e}rczi, Kir{\'a}ly, and Szabo~\cite{BKS24} posed a varient of \MC\ where the $k$ terminals are not fixed, but rather the $i$th terminal can be chosen from a subset $S_i$ of the vertices provided as part of the input. They study the existence (or non-existence) of approximation algorithms for many versions of this question. Of note, a number of their approximation algorithms use an optimal approximation algorithm for \MC\ as a black-box. In particular, Theorem~\ref{thm:main} automatically improves the quantitative guarantees of many of their algorithms.

Effort has also been put into solving \MC\ exactly for special families of graphs. As previously mentioned, Dahlhaus et al.~\cite{DJPSY92,DJPSY94} showed that \MC\ can be efficiently solved for the class of planar graphs. More recently, Hirai~\cite{H18} developed a theory of submodular optimization on a rather general set of lattices, showing as a corollary that problems like \MC\ and $0$-Extension can be solved exactly on ``orientable modular'' graphs.

A notable open question related to \MC\ is that of \emph{kernelization}, where one seeks to efficiently replace the input (unweighted) graph with a small graph such that (1) all $k$ terminals are preserved in the smaller graph and (2) the optimal \MC\ in both graphs is identical. A natural target for the size of the kernel is the \emph{sums of degrees} $D$ of the terminals as one can always find a multiway cut using at most $D$ edges.\footnote{Many works in this field use $s$ instead of (our) $k$ and $k$ instead of (our) $D$.} Kratsch and Wahlstr{\"o}m~\cite{KW20} can efficiently construct via a randomized algorithm a kernel of size $O(D^{k+1})$, which is polynomial for fixed $k$. For general $k$, recent work by Wahlstr{\"o}m~\cite{W22} proved that one can efficiently construct such a kernel of \emph{quasipolynomial} size $2^{O(\log^4 D)}$ also via a randomized algorithm. See also the follow-up work of Wahlstr{\"o}m~\cite{W24} which constructs kernels for variants of \MC.

\subsection{Organization}

The rest of the paper is organized as follows. In the next section we define the basic families of rounding schemes used in the paper, including the new family of Generalized Kleinberg-Tardos (GKT) rounding schemes, and analyze each one of them. The analysis of GKT is novel to the best of our knowledge. We also extend the analysis previously given for some of the other rounding schemes. In Section~\ref{sec:further} we obtain bounds on the density functions of the GKT and IT rounding schemes when only a prefix of a simplex point is given. (See the discussion in Section~\ref{sub-verify}.) In Section~\ref{sec:minimax} we discuss the process we used for the discovery of large mixtures of rounding schemes. In Section~\ref{sec:new-schemes} we describe the main features of the newly discovered rounding schemes. In Section~\ref{sec:verification} we consider the rigorous verification of the approximation ratios claimed. In Section~\ref{sec:discussion}, we discuss how our new results relate to the quest of finding truly optimal rounding schemes for \MC. We end in Section~\ref{sec:concl} with some concluding remarks and open problems. Further material appears in appendices.

\section{Rounding Schemes and Their Analysis}\label{sec:rounding}

In this section we describe the different rounding schemes used by our algorithms. We start with some general definitions and observations that apply to all schemes and then consider each scheme separately.

All the rounding schemes $R$ considered in the paper, except for Exponential Clocks (EC) which is discussed briefly, are of the following form. The scheme chooses a random sequence $\sigma(1),\sigma(2),\ldots$ of terminals and a random sequence $t_1,t_2,\ldots$ of \emph{thresholds}. In the $i$-th round, the scheme assigns all yet unassigned simplex points $(u_1,u_2,\ldots,u_k)$ with $u_{\sigma(i)}\ge t_i$ to terminal~$\sigma(i)$. The schemes differ in the way the random sequences of terminals and thresholds are chosen. In most schemes $\sigma(1),\sigma(2),\ldots,\sigma(k)$ is actually a permutation and $t_k=0$. That means that all simplex points that are yet unassigned at the beginning of the $k$-th round are assigned to the `default' terminal $\sigma(k)$.

\begin{definition}[Cut probabilities]
    For a rounding scheme $R$, 
    \begin{itemize}
      \renewcommand\labelitemi{}  
      \item Let $P^R_{k,\eps}(u_1,u_2,\ldots,u_k)$ be the probability that the $(1,2)$-aligned edge of length~$\eps$ connecting $(u_1,u_2,\ldots,u_k)$ and $(u_1-\eps,u_2+\eps,\ldots,u_k)$ is cut by $R$, i.e., the two endpoints of the edge are assigned by~$R$ to different terminals.
      \item Let $P^{R,1}_{k,\eps}(u_1,u_2,\ldots,u_k)$ be the probability that $(u_1,u_2,\ldots,u_k)$ is assigned by~$R$ to terminal~$1$ and that $(u_1-\eps,u_2+\eps,\ldots,u_k)$ is assigned, at a later stage, to a different terminal.
      \item Let $P^{R,2}_{k,\eps}(u_1,u_2,\ldots,u_k)$ be the probability that $(u_1-\eps,u_2+\eps,\ldots,u_k)$ is assigned by~$R$ to terminal~$2$ and that $(u_1,u_2,\ldots,u_k)$ is assigned, at a later stage, to a different terminal.
    \end{itemize}    
\end{definition}

For all rounding schemes~$R$ considered in this paper we have:
\[ P^R_{k,\eps}(u_1,u_2,\ldots,u_k) \EQ P^{R,1}_{k,\eps}(u_1,u_2,\ldots,u_k) + P^{R,2}_{k,\eps}(u_1,u_2,\ldots,u_k) \;. \]
This is because if $(u_1,u_2,\ldots,u_k)$ and $(u_1-\eps,u_2+\eps,\ldots,u_k)$ are still unassigned, they can only be assigned together to terminals other than~$1$ and~$2$. In other words, only terminals $1$ or~$2$ can cut a $(1,2)$-aligned edge. The `later stage' requirement is added to ensure that the event in which $(u_1,u_2,\ldots,u_k)$ is assigned to terminal~$1$ and $(u_1-\eps,u_2+\eps,\ldots,u_k)$ is assigned to terminal~$2$ is not counted twice. This subtlety is important only for the Generalized Kleiberg-Tardos schemes (GKT).

Note also that by the symmetry of our rounding schemes we have that
\[ P^{R,2}_{k,\eps}(u_1,u_2,\ldots,u_k) \EQ P^{R,1}_{k,\eps}(u_2+\eps,u_1-\eps,\ldots,u_k)\;, \]
and that $P^R_{k,\eps}(u_1,u_2,\ldots,u_k)$ is symmetric in $u_3,\ldots,u_k$, i.e., it does not depend on the order of these arguments. 

\begin{definition}[Cut densities] 
\[ d^R_k(u_1,\ldots,u_k) \EQ \lim_{\eps\to0}\frac{P^R_{k,\eps}(u_1,u_2,\ldots,u_k)}{\eps} \quad,\quad
d^{R,i}_k(u_1,\ldots,u_k) \EQ \lim_{\eps\to0}\frac{P^{R,i}_{k,\eps}(u_1,u_2,\ldots,u_k)}{\eps} \;.\]
\end{definition}

As before, we have:
\[ d^R_{k}(u_1,u_2,\ldots,u_k) \EQ d^{R,1}_{k}(u_1,u_2,\ldots,u_k) + d^{R,2}_{k}(u_1,u_2,\ldots,u_k) \;. \]

If $d^{R,1}_{k}(u_2,u_1,\ldots,u_k)$ is continuous, which holds for most of the schemes we consider, then:
\[ d^{R,2}_{k}(u_1,u_2,\ldots,u_k) \EQ d^{R,1}_{k}(u_2,u_1,\ldots,u_k) \;.\]
Thus, to compute the cut density $d^R_{k}(u_1,u_2,\ldots,u_k)$ of a scheme~$R$ we can concentrate on computing $d^{R,1}_{k}(u_1,u_2,\ldots,u_k)$ and a corresponding formula for $d^{R,2}_{k}(u_1,u_2,\ldots,u_k)$ would follow.

All the rounding schemes we use are parametrized by a continuous random variable~$T$ assuming values in $[0,1]$. \footnote{Strictly speaking, the random variable~$T$ is not allowed to assume the value~$0$, as this may lead to a vertex $\be_i$ of the simplex being assigned to a terminal other than terminal~$i$. As we consider continuous random variables, this happens with probability~$0$.} We let $f:[0,1]\to[0,\infty)$ be the density function (pdf) of this random variable, and $F(x)=\int_{0}^x f(t)\,dt$ be the corresponding cumulative distribution function (cdf). Thus $\Pr[T\le t]=F(t)$. We also let $F(u_1,u_2)=F(u_2)-F(u_1)$, for $u_1\le u_2$, so that $\Pr[t_1<T\le t_2] = F(t_1,t_2)$.

We next consider the individual rounding schemes used, starting with the newly introduced generalized version of the Kleinberg-Tardos rounding scheme. We next consider the Independent Thresholds (IT) rounding scheme which was considered before but for which we need to introduce a more delicate analysis. Finally, we consider the Single Threshold (ST) and Descending Thresholds (DT) rounding schemes.

\subsection{Generalized Kleinberg-Tardos (GKT)}\label{sub:GKT}

Kleinberg and Tardos \cite{KlTa02} introduced their rounding scheme to obtain a 2-approximation algorithm for the uniform metric labeling problem. It uses the uniform random variable on $[0,1]$. Buchbinder, Naor, and Schwartz \cite{BNS18} showed that it is equivalent to their Exponential Clocks (EC) rounding scheme. Interestingly, none of the previous papers on the Multiway Cut problem considered using the rounding scheme of Kleinberg and Tardos with a non-uniform random variable, even though non-uniform random variables were used in all other rounding schemes. We show that using the rounding scheme of Kleinberg and Tardos with a non-uniform random variable can lead to significantly improved approximation ratios for the Multiway Cut problem.

A Generalized Kleinberg-Tardos scheme KT$(f)$, where $f$ is a density function of a continuous random variable with values in $[0,1]$, works as follows. It repeatedly chooses a terminal $i\in[k]$ uniformly at random, and a random threshold $t$ chosen according to the density function~$f$, independent of all previous choices. If a simplex point $\bu=(u_1,\ldots,u_k)$ is yet unassigned and $u_{i}\ge t$ then~$\bu$ is assigned to terminal~$i$. This goes on until all simplex points, or until all simplex points that appear in the embedding of the input graph into the simplex, are assigned to terminals. To ensure the termination of KT$(f)$ we require that $F(\frac{1}{k})>0$ (since for any simplex point $\bu=(u_1,\ldots,u_k)$, there exists some $i \in [k]$ such that $u_i \geq 1/k$).

KT$(f)$ differs from all other rounding schemes considered in this paper in that it does not choose a random permutation of the terminals. Instead, it samples terminals with repetitions. 

We next provide an analysis of KT$(f)$ generalizing the analysis of \cite{KlTa02} and \cite{BNS18}.

\begin{lemma}
    \[P^{\KT(f),1}_{k,\eps}(u_1,u_2,\ldots,u_k) \EQ\]
    \[\frac{ F(u_1-\eps,u_1) }{ F(u_1) + F(u_2+\eps) + \sum_{i=3}^k F(u_i) }\left(1 - \frac{F(u_1-\eps)}{F(u_1-\eps)+F(u_2+\eps)+\sum_{i=3}^k F(u_i)} \right)\;.\]
\end{lemma}
\begin{proof}
    Let $\bu=(u_1,u_2,\ldots,u_k)$ and $\bv=(u_1-\eps,u_2+\eps,\ldots,u_k)$. A schematic description of the computation of $P^{R,1}_{k,\eps}(u_1,u_2,\ldots,u_k)$ is given in Figure~\ref{fig:KT}. We let $c(\bu)$ and $c(\bv)$ be the indices of the terminals assigned to $\bu$ and $\bv$, or $\bot$ if no terminal is assigned yet. Initially $c(\bu)=c(\bv)=\bot$. In each round, KT$(f)$ chooses a random terminal~$i$, uniformly at random, and a random threshold~$t$ according to the density~$f$. One of the following four things can happen:
    \begin{enumerate}
        \item $\bu$ is assigned to terminal~$1$ while $\bv$ remains unassigned. This happens if $i=1$ and $t\in(u_1-\eps,u_1]$, hence with probability $\frac1k F(u_1-\eps,u_1)$.
        \item $\bv$ is assigned to terminal~$2$ while $\bu$ remains unassigned. This happens if $i=2$ and $t\in(u_2,u_2+\eps]$, hence with probability $\frac1k F(u_2,u_2+\eps)$.
        \item $\bu$ and $\bv$ are both assigned to the same terminal. This happens either if $i=1$ and $t\in(0,u_1-\eps]$, or if $i>1$ and $t\in (0,u_i]$, hence with probability $\frac1k F(u_1-\eps) + \frac1k \sum_{i=2}^k F(u_i)$.
        \item $\bu$ and $\bv$ both remain unassigned. This happens with the complementary probability.
    \end{enumerate}
    Thus, the probability that the first transition out of the state $c(\bu)=c(\bv)=\bot$ is to the state $c(\bu)=1$ and $c(\bv)=\bot$, shown in the diagram in orange, is
    \[ \frac{ \frac1k F(u_1-\eps,u_1) }
    { \frac1k F(u_1{-}\eps,u_1) + \frac1k F(u_2,u_2{+}\eps) + \frac1k F(u_1{-}\eps) + \frac1k \sum_{i=2}^k F(u_i)} \!\EQ\!
    \frac{ F(u_1-\eps,u_1) }{ F(u_1) + F(u_2{+}\eps) + \sum_{i=3}^k F(u_i) }\;.\]
    If the first transition out of $c(\bu)=c(\bv)=\bot$ is to a different state, then the event whose probability the expression $P^{\KT(f),1}_{k,\eps}(u_1,u_2,\ldots,u_k)$ is computing does not happen, as we require $\bu$ to be assigned to terminal~$1$ \emph{before} $\bv$ is assigned to a different terminal.

    Next, if we are in the orange state $c(\bu)=1$ and $c(\bv)=\bot$, three possible things can happen in each following round:
    \begin{enumerate}
        \item $\bv$ is also assigned to terminal~$1$. This happens if $i=1$ and $t\in(0,u_1-\eps]$, hence with probability $\frac1k F(u_1-\eps)$.
        \item $\bv$ is assigned to a terminal other~$1$. This happens if $i=2$ and $t\in(0,u_2+\eps]$ or if $i>2$ and $t\in(0,u_i]$, hence with probability $\frac1k F(u_2+\eps) + \frac1k \sum_{i=3}^k F(u_i)$.
        \item $\bv$ remains unassigned. This happens with the complementary probability.
    \end{enumerate}
    Thus, the probability that the first transition out of the orange state $c(\bu)=1, c(\bv)=\bot$ is to the red state $c(\bu)=1,c(\bv)> 1$ is
    \[ 1 - \frac{ \frac1k F(u_1-\eps) }{ \frac1k F(u_1-\eps) + \frac1k F(u_2+\eps) + \frac1k \sum_{i=3}^k F(u_i)} \EQ 1 - \frac{ F(u_1-\eps) }{ F(u_1-\eps) + F(u_2+\eps) + \sum_{i=3}^k F(u_i)}\;.\]
    Combining these two probabilities, we get the expression appearing in the statement of the lemma.
\end{proof}

\begin{figure}[t]
    \centering
    \begin{tikzpicture}
\node[ellipse,draw,thick,align=center] at (0,0) (top) {$c(\bu)=\bot$\\$c(\bv)=\bot$};
\node[ellipse,draw,thick,align=center,orange,text=black] at (-5,-3) (twoA) {$c(\bu)=1$\\$c(\bv)=\bot$};
\node[ellipse,draw,thick,align=center] at (-2,-3) (twoB) {$c(\bu)=\bot$\\$c(\bv)=2$};
\node[ellipse,draw,thick,align=center,minimum height=43] at (3.5,-3) (twoC) {$c(\bu)=c(\bv)$};
\node[ellipse,draw,thick,align=center] at (-6.5,-6) (threeA) {$c(\bu)=1$\\$c(\bv)=1$};
\node[ellipse,draw,thick,align=center,red,text=black] at (-3,-6) (threeB) {$c(\bu)=1$\\$c(\bv)>1$};

\draw[->] (top.west) to[out=180, in=90, looseness=3](top.north);
\draw[->] (twoA.west) to[out=180, in=90, looseness=3](twoA.north);

\draw[->] (top.225)--(twoA) node[near start, above left, blue, scale=0.8] {$\dfrac{1}{k}F(u_1-\eps, u_1)$};
\draw[->] (top.270)--(twoB) node[midway, below right, blue, scale=0.8] {$\dfrac{1}{k}F(u_2, u_2+\eps)$};
\draw[->] (top.315)--(twoC) node[midway, blue, scale=0.8, above right] {$\dfrac{1}{k}F(u_1-\eps) + \dfrac{1}{k}\sum_{i=2}^k F(u_i)$};

\draw[->] (twoA.225)--(threeA) node[midway, left, blue, scale=0.8] {$\dfrac{1}{k}F(u_1-\eps)$};
\draw[->] (twoA.315)--(threeB) node[near end, blue, scale=0.8, above right] {$\dfrac{1}{k}F(u_2+\eps) + \dfrac{1}{k}\sum_{i=3}^k F(u_i)$};
\end{tikzpicture}
    \caption{Analysis of KT$(f)$ - The computation of $P^{\KT(f),1}_{k,\eps}(u_1,u_2,\ldots,u_k)$.}
    \label{fig:KT}
\end{figure}
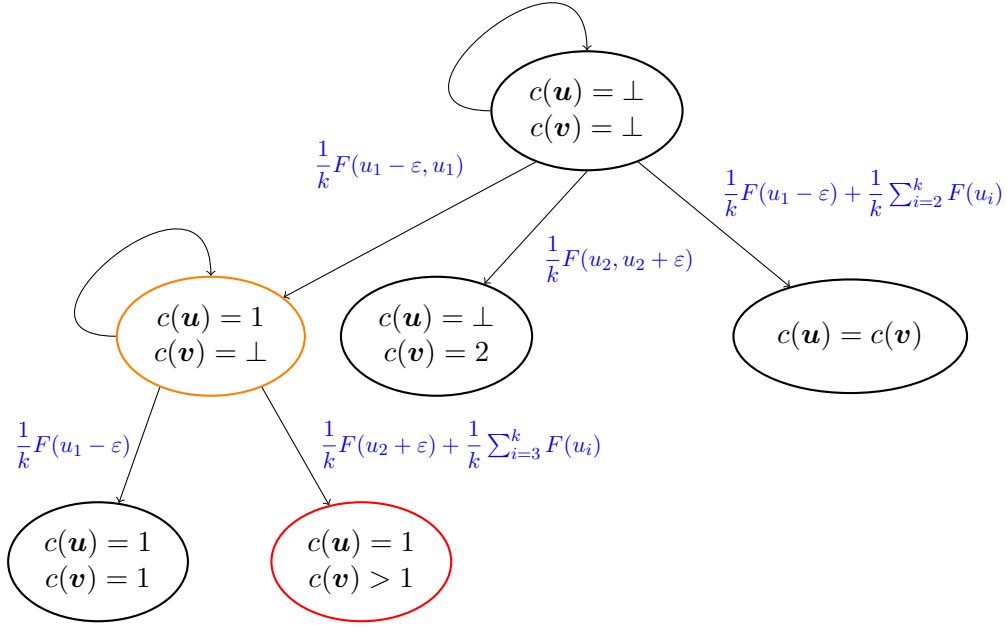

As an immediate corollary, we get:

\begin{corollary}\label{C-KT}
    $\displaystyle d^{\KT(f),1}_{k}(u_1,u_2,\ldots,u_k) \EQ
    \frac{ f(u_1) }{ \sum_{i=1}^k F(u_i) }\left(1 - \frac{F(u_1)}{\sum_{i=1}^k F(u_i)} \right)\;.$
\end{corollary}

For a uniform random variable on $[0,1]$ we have $f(u)=1$ and $F(u)=u$. Hence $\sum_{i=1}^k F(u_i)=\sum_{i=1}^k u_i=1$. Thus $d^{\KT(f),1}_{k}(u_1,u_2,\ldots,u_k)=1-u_1$ and $d^{\KT(f)}_{k}(u_1,u_2,\ldots,u_k)=2-u_1-u_2$, matching the analysis of \cite{BNS18} for the uniform case, which also matches the performance of Exponential Clocks (EC).

\subsection{Independent Thresholds (IT)}\label{sub:IT}

The Independent Thresholds (IT) rounding scheme was introduced by Karger et al.\ \cite{KKSTY04}. For general~$k$, they used it with a uniform random variable on the interval $(0,\frac{6}{11}]$. In this section we extend and generalize the analysis of IT schemes given in \cite{KKSTY04}.

An IT$(f)$ scheme, where $f$ is a density function of a continuous random variable that attains values in $[0,1]$, chooses a random permutation $\sigma$ on the $k$ terminals. It also chooses $k-1$ \emph{independent} thresholds $t_1,t_2,\ldots,t_{k-1}$ according to the distribution~$f$. For $i=1,2,\ldots,k-1$, if a simplex point $\bu=(u_1,\ldots,u_k)$ is yet unassigned and $u_{\sigma^{-1}(i)}\ge t_i$ then $\bu$ is assigned to terminal~$\sigma^{-1}(i)$. At the end, all unassigned simplex points are assigned to terminal $\sigma^{-1}(k)$. (It is convenient to use $\sigma^{-1}$ in the definition, as then $\sigma(i)$ is the number of the round in which terminal~$i$ is considered.)

The following lemma is a slight modification of a formula appearing in \cite{KKSTY04}. Here $S_k$ denotes the set of all permutations on $[k]$.

\begin{lemma}
    $\displaystyle P^{\IT(f),1}_{k,\eps}(u_1,u_2,\ldots,u_k) \EQ (F(u_1)-F(u_1-\eps)) \cdot \frac{1}{k!}\sum_{\substack{\sigma\in S_k\\\sigma(1)\ne k}}  \prod_{i:\sigma(i)<\sigma(1)}(1-F(u_i+\eps_i)) \;,$
    where $\eps_2=\eps$ and $\eps_i=0$ for $i\ge 3$.
\end{lemma}

\begin{proof}
    The event happens if and only if the random threshold $t_{\sigma(1)}$ chosen for terminal~$1$ falls in the interval $(u_1-\eps,u_1]$, and the threshold $t_{\sigma(i)}$ of each terminal $i$ that appears before terminal~$1$ in the random permutation~$\sigma$ chosen by the scheme falls in the interval $(u_i,1]$. For terminal~$2$, the threshold should actually be in $(u_2+\eps,1]$ as otherwise $(u_1-\eps,u_2+\eps,\ldots,u_k)$ will be assigned to terminal~$2$ before $(u_1,u_2,\ldots,u_k)$ is assigned to terminal~$1$. In addition to that, terminal~$1$ should not be the last terminal in the permutation, hence the condition $\sigma(1)\ne k$. 
\end{proof}

We note that the effect of the condition $\sigma(1)\ne k$ becomes negligible for large or unbounded values of~$k$. However, it is important when considering small values of~$k$.

As an immediate corollary of lemma above we get:
\begin{corollary}
$\displaystyle d_k^{\IT(f),1}(u_1,u_2,\ldots,u_k) \EQ f(u_1) \cdot \frac{1}{k!}\sum_{\substack{\sigma\in S_k\\\sigma(1)\ne k}}  \prod_{i:\sigma(i)<\sigma(1)}(1-F(u_i)) \;.$
\end{corollary}

We now go beyond the analysis of $\IT(f)$ given by \cite{KKSTY04} and obtain efficient ways of evaluating or bounding $d_k^{\IT(f),1}(u_1,u_2,\ldots,u_k)$. 
In particular, relations between $d_k^{\IT(f),1}(u_1,u_2,\ldots,u_k)$ and \emph{elementary symmetric polynomials} are revealed. In Section~\ref{sub:further-IT} we go further and consider the asymptotic behavior of $d_k^{\IT(f),1}(u_1,u_2,\ldots,u_k)$ as $k\to\infty$.

Let $A=\{i\mid \sigma(i)<\sigma(1)\}\subseteq[2,k]$, \footnote{If $j$ and $k$ are integers, we let $[j,k]=\{j,j+1,\ldots,k\}$ and $[k]=\{1,2,\ldots,k\}$.} where~$\sigma$ is a random permutation chosen uniformly at random. Note that~$A$ is a random set, that~$|A|$ is uniformly distributed in $[k-1]$, and that by symmetry all subsets $A\subseteq[2,k]$ of a given size are equally likely. We thus have:
\[ d_k^{\IT(f),1}(u_1,\ldots,u_k) \EQ f(u_1) \cdot \frac{1}{k} \sum_{a=0}^{k-2} \frac{1}{\binom{k-1}{a}} \sum_{\substack{A\subseteq[2,k]\\ |A|=a}} \prod_{i\in A}(1-F(u_i)) \;,\]
where the upper limit of the sum is $k-2$, and not $k-1$, as we want to exclude permutations in which terminal~$1$ is the last terminal.

\begin{definition}[Elementary symmetric functions]
    Let 
    \[ e_k(x_1,\ldots,x_n) \EQ \sum_{\substack{A\in[n]\\ |A|=k}}\prod_{i\in A}x_i 
    \quad,\quad E_k(x_1,\ldots,x_n) \EQ \frac{e_k(x_1,\ldots,x_n)}{\binom{n}{k}}  \]
    be the $k$-th elementary symmetric function, and the normalized $k$-th elementary symmetric function, in the variables $x_1,x_2,\ldots,x_n$.
\end{definition}

Given $x_1,\ldots,x_n$, the elementary symmetric functions $e_k(x_1,\ldots,x_n)$, for $k=0,1,\ldots,n$, can be evaluated by computing the coefficients of the univariate polynomial 
\[\prod_{i=1}^n(1-tx_i) \EQ \sum_{k=0}^n e_k(x_1,\ldots,x_n)t^k\;.\]
Using FFT, the product of two univariate polynomials of degree~$n$ can be computed in time $O(n\log n)$. Hence the product of~$n$ degree-$1$ polynomials, and the values of all the elementary symmetric polynomials on~$n$ variables, can be computed in $O(n\log^2n)$ time.

\begin{corollary}\label{C-elementary}
   Let $y_i=1-F(u_i)$, for $i\in[k]$, then
   \[ d_k^{\IT(f),1}(u_1,\ldots,u_k) \EQ f(u_1)\cdot \frac{1}{k}\sum_{a=0}^{k-2} \frac{e_a(y_2,\ldots,y_k)}{\binom{k-1}{a}}
   \EQ f(u_1)\cdot \frac{1}{k}\sum_{a=0}^{k-2} E_a(y_2,\ldots,y_k)\;.\]
\end{corollary}

Corollary~\ref{C-elementary} provides an efficient way of computing $d_k^{\IT(f),1}(u_1,\ldots,u_k)$, and $d_k^{\IT(f)}(u_1,\ldots,u_k)$, as the elementary symmetric polynomials can be evaluated efficiently as explained above.

Using a simple decomposition property of the elementary symmetric polynomials we also get the following corollary that we use in Section~\ref{sec:further}.

\begin{corollary}\label{C-IT-decompose}
   Let $y_i=1-F(u_i)$, for $i\in[k]$, then for any $\ell \in [k]$ we have
        \[ d^{\IT(f),1}_k(u_1,u_2,\ldots,u_k) \EQ f(u_1) \cdot \frac{1}{k} \sum_{\substack{0\le a<\ell\\0\le b\le k-\ell\\a+b\le k-2}} \frac{1}{\binom{k-1}{a+b}} e_a(y_2,\ldots,y_\ell)\cdot e_b(y_{\ell+1},\ldots,y_k) \;.\]
\end{corollary}

The following lemma, which seems to be new, gives an integral formula for $d_k^{\IT(f),1}(u_1,\ldots,u_k)$ with a simple intuitive derivation. We do not rely on this formula in our analysis so its proof is deferred to Appendix~\ref{A-IT-integral}.

\begin{lemma}\label{L-IT-int-formula}
    For every $k\ge 3$ we have:
    \[ d_k^{\IT(f),1}(u_1,\ldots,u_k)
    \LE f(u_1)\cdot \left( \int_0^1 \prod_{i=2}^k(1-tF(u_i))\,dt - \frac{1}{k}\prod_{i=2}^k (1-F(u_i)) \right)\;.\]
\end{lemma}

\subsection{Single Threshold (ST)}\label{sub:ST}

The Single Threshold (ST) rounding scheme was introduced by C{\u{a}}linescu, Karloff, and Rabani \cite{CKR00} who used it with a uniform random variable assuming values in $[0,1]$. 

An ST$(f)$ scheme, where $f$ is a density function of a continuous random variable that attains values in $[0,1]$, chooses a random permutation $\sigma$ on the $k$ terminals and a \emph{single} threshold $t$ chosen according to the distribution~$f$. For $i=1,2,\ldots,k-1$, if a simplex point $\bu=(u_1,\ldots,u_k)$ is unassigned and $u_{\sigma(i)}\ge t$ then $\bu$ is assigned to terminal~$\sigma(i)$. At the end, all unassigned simplex points are assigned to terminal $\sigma(k)$.

\begin{lemma}\label{ST-density-1}
        $\displaystyle d^{\ST(f),1}_k(u_1,u_2,\ldots,u_k) \EQ 
        \begin{ccases}
(1-\frac{1}{k})f(u_1) & \text{if } u_1>u_2,\ldots,u_k\\[8pt]
\displaystyle\frac{f(u_1)}{|\{i\in[k]\mid u_i\ge u_1\}|}          & \text{otherwise}.
\end{ccases} $
\end{lemma}

\begin{proof}
    Assume at first that $u_1\le \max\{u_2,\ldots,\allowbreak u_k\}$. We then have:
    \[ P^{\ST(f),1}_k(u_1,u_2,\ldots,u_k) \EQ \int_{u_1-\eps}^{u_1} \frac{f(t)}{|\{i\in[k]\mid u_i\ge t\}|}\,dt \;.\]
    This follows as we need $t\in(u_1-\eps,u_1]$ and terminal~$1$ needs to appear in the random permutation before every terminal $i>1$ for which $u_i\ge t$, otherwise this terminal would capture both endpoints of the edge. (Note that this condition also implies that teminal~$1$ is not the last terminal in the permutation.) Taking the limit, we get the claimed expression for $d^{\ST(f),1}_k(u_1,u_2,\ldots,u_k)$.

    If $u_1>u_2,\ldots,u_k$, we need $t\in(u_1-\eps,u_1]$ and that terminal~$1$ is not the last terminal in the permutation, which happens with probability $1-\frac1k$. This again gives the claimed expression.
\end{proof}

It is interesting to note that $d^{\ST(f),1}_k(u_1,u_2,\ldots,u_k)$ is not a continuous function. For example, if $u_1<u_2$, then $d^{\ST(f),1}(u_1,u_2,u_1)=\frac13 f(u_1)$, while $d^{\ST(f),1}(u_1,u_2,u_1-\delta)=\frac12 f(u_1)$, for every $\delta>0$. The expressions for $d^{\ST(f),2}_k(u_1,u_2,\ldots,u_k)$ are thus similar, but not completely identical: \footnote{The expressions will become identical if we consider an edge with endpoints $(u_1-\frac{\eps}{2},u_2+\frac{\eps}{2},u_3,\ldots,u_k)$ and $(u_1+\frac{\eps}{2},u_2-\frac{\eps}{2},u_3,\ldots,u_k)$, but they become slightly more complicated, with both strict and non-strict inequalities.}

\begin{lemma}\label{ST-density-2}
        $\displaystyle d^{\ST(f),2}_k(u_1,u_2,\ldots,u_k) \EQ 
        \begin{ccases}
(1-\frac{1}{k})f(u_2) & \text{if } u_2\ge u_1,u_3,\ldots,u_k\\[8pt]
\displaystyle\frac{f(u_2)}{1+|\{i\in[k]\mid u_i> u_2\}|}          & \text{otherwise}.
\end{ccases} $
\end{lemma}

In the experimental discovery of our rounding schemes we only relied on the following potentially weaker corollary. The stronger bounds of Lemmas~\ref{ST-density-1} and~\ref{ST-density-2} are used, however, to speed up the verification process. 

\begin{corollary}\label{ST-density}
    If $u_1\le u_2$, then $d^{\ST(f)}_k(u_1,u_2,\ldots,u_k) \LE \frac{1}{2}f(u_1)+(1-\frac1k)f(u_2)$.
\end{corollary}

It follows immediately from this corollary that the approximation ratio achieved by the algorithm of \CKRal\ is at most $\frac{3}{2}-\frac{1}{k}$, as they just use ST$(f)$ with $f(x)=1$, for $x\in[0,1]$.

\subsection{Descending Thresholds (DT)}\label{sub:DT}

The Descending Threshold (DT) rounding scheme was introduced by Sharma and Vondr{\'a}k \cite{SV14} who used it with a random variable that is uniform in the interval $[0,\frac{6}{11}]$. We again use it with more general random variables. An DT$(f)$ scheme, where $f$ is a density function of a continuous random variable that attains values in $[0,1]$, chooses for each terminal~$i$ a random threshold~$t_i$ according to~$f$. All choices are independent. It then lets $\sigma$ be the permutation that sorts the random thresholds in \emph{decreasing} order, i.e., $t_{\sigma(1)}> t_{\sigma(2)}>\cdots> t_{\sigma(k)}$. (We assume that~$f$ is finite in $[0,1]$, so the probability of a tie is~$0$.) For $i=1,2,\ldots,k-1$, if a simplex point $\bu=(u_1,u_2,\ldots,u_k)$ was not assigned yet to a terminal and $u_{\sigma(i)}\ge t_{\sigma(i)}$, then $\bu$ is assigned to terminal~$\sigma(i)$. At the end, all points that were not assigned are assigned to terminal~$\sigma(k)$.

Both IT$(f)$ and DT$(f)$ choose independent random thresholds $t_i$ for each terminal according to~$f$. The main difference is that IT$(f)$ also chooses a random permutation independent of all thresholds while DT$(f)$ uses the random thresholds to define the permutation. 

Recall that $F(u)=\int_{0}^u f(x)\,dx$ is the cumulative probability function of a random variable with density function~$f(x)$ and that $F(u_1,u_2)=F(u_2)-F(u_1)$, for $u_1\le u_2$, is the probability that the random variable attains a value in $(u_1,u_2]$.

Generalizing expressions appearing in Sharma and Vondr\'ak \cite{SV14} we obtain:

\begin{lemma}\label{DT-prob}
   \[ P^{\DT(f),1}_{k,\eps}(u_1,u_2,\ldots,u_k) \!\EQ\! \int_{u_1-\eps}^{u_1} f(t) \left(\prod_{\substack{i\ne 1\\i:u_i>t}}(1-F(t,u_i))-\prod_{i\ne 1}F(\max\{t,u_i\},1) \!\right)dt\;.\]
\end{lemma}

\begin{proof}
    For terminal~$1$ to cut the edge, three things need to happen: (1) $t_1\in (u_1-\eps,u_1]$; (2) For every $i\ne 1$ we need that $t_i\notin (t_1,u_i]$, as otherwise terminal~$i$ would be considered before terminal~$1$ and would capture both endpoints of the edge; (3) Terminal~$1$ is not the last in the permutation, i.e., there exists $i\ne 1$ for which $t_1\ge t_i$.

    Conditioning on $t_1=t$, where $t\in(u_1-\eps,u_1]$, the probability that condition (2) holds, i.e., that $t_i\notin (t,u_i]$ for every $i\ne 1$, is $\prod_{\substack{i\ne 1\\i:u_i>t}}(1-F(t,u_i))$, as all the thresholds are chosen independently. (Note that if $u_i\le t$, then the condition $t_i\notin (t,u_i]$ holds automatically.) The probability that condition (2) holds while condition (3) does not is $\prod_{i\ne 1}F(\max\{t,u_i\},1)$. 
\end{proof}

As an immediate corollary, we get:

\begin{corollary}\label{DT-density}
    If $f(u)$ is continuous, then
    \[ d^{\DT(f),1}_k(u_1,u_2,\ldots,u_k) \EQ f(u_1) \cdot \left(\prod_{i:u_i\ge u_1}(1-F(u_1,u_i))-\prod_{i\ne 1}F(\max\{u_1,u_i\},1) \right)\;.\]
\end{corollary}

Note that $F(u_1,u_i)=F(u_i)-F(u_1)$, if $u_1<u_i$, and that $F(\max\{u_1,u_i\},1) = 1-F(\max\{u_1,u_i\})$.
Also note that if $u_i=u_1$, then $1-F(u_1,u_1)=1$, hence the condition $u_i\ge u_1$ can be replaced by $u_i>u_1$, and it also not necessary to require that $i\ne 1$. For the same reason, $d^{\DT(f),1}_k(u_1,u_2,\ldots,u_k)$ is a continuous function of $u_1,u_2,\ldots,u_k$.

The term $\prod_{i\ne 1}F(\max\{u_1,u_i\},1)$ is significant only when $k$ is small. When $k$ is large, or unbounded, we use the upper bound on $d^{\DT(f),1}_k(u_1,u_2,\ldots,u_k)$ obtained by ignoring this term.

\section{Further Analysis of GKT and IT}\label{sec:further}

In this section we present further analysis of the GKT and IT rounding schemes. In particular, we obtain bounds on $d^R_k(u_1,u_2,\ldots,u_k)$, where $R=\KT(f)$ or $R=\IT(f)$, when a prefix $(u_1,u_2,\ldots,u_\ell)$ of the simplex point, where $2\le \ell<k$, is given, and where the suffix $(u_{\ell+1},\ldots,u_k)$ may be arbitrary. These bounds play a crucial role in the discovery of our improved rounding schemes and in their rigorous verification. To obtain these bounds, we need to make some further assumptions on the distributions used. For GKT schemes, we assume that the density function $f(x)$ is constant, and positive, in the interval $[0, \alpha]$, for some $\alpha>0$. For IT schemes, we assume that the cumulative density function $F(x)$ is convex in an interval $[0, \alpha]$.

It is convenient to make the following definition:

\begin{definition}[Cut density of prefixes]
    Let $R$ be a rounding scheme. If $2\le \ell\le k$ and $u_1,u_2,\ldots,u_\ell$ are such that $u_1,u_2,\ldots,u_\ell\ge 0$ and $\sum_{i=1}^\ell u_i\le 1$, let
\[ d^{R,1}_k(u_1,u_2,\ldots,u_\ell) \EQ \max_{\substack{u_{\ell+1},\ldots,u_k\ge 0\\ \sum_{i=1}^k u_i=1}} d^{R,1}_k(u_1,u_2,\ldots,u_\ell,u_{\ell+1},\ldots,u_k) \;.\]
\end{definition}
$d^{R,2}_k(u_1,u_2,\ldots,u_\ell)$ and $d^{R}_k(u_1,u_2,\ldots,u_\ell)$ are defined analogously. It follows immediately from the above the definition that $d^{R,1}_k(u_1,u_2,\ldots,u_\ell)\le d^{R,1}_{k+1}(u_1,u_2,\ldots,u_\ell)$ as in the maximum defining $d^{R,1}_{k+1}(u_1,u_2,\ldots,u_\ell)$ we are free to choose $u_{k+1}=0$ and $d^{R,1}_{k+1}(u_1,u_2,\ldots,u_k,0) = d^{R,1}_{k}(u_1,u_2,\ldots,u_k)$. We thus define:

\begin{definition}[Asymptotic Cut density]\label{D:inf}
    Let $R$ be a rounding scheme. If $2\le \ell$, $0\le u_1,u_2,\ldots,u_\ell$ and $\sum_{i=1}^\ell u_i\le 1$, let
\[ d^{R,1}_\infty(u_1,u_2,\ldots,u_\ell) \EQ \lim_{k\to\infty} d^{R,1}_k(u_1,u_2,\ldots,u_\ell) \EQ \sup_{k\ge 2}\; d^{R,1}_k(u_1,u_2,\ldots,u_\ell)\;.\]
\end{definition}

\subsection{Further Analysis of GKT}\label{sub:further-GKT}

The bound for GKT is relatively straightforward. 

\begin{lemma}\label{L-KT-prefix}
    Suppose that $f$ is constant and nonzero on $[0, \alpha]$, for some $\alpha > 0$. If $u_{\ell+1},\ldots u_k\in [0, \alpha]$, where $2\le\ell\le k$, then 
    \[ d^{\KT(f),1}_k(u_1,u_2,\ldots,u_k) \EQ \]
    \[\frac{ f(u_1) }{ \sum_{i=1}^\ell F(u_i) + f(0)\left(1-\sum_{i=1}^\ell u_i\right) }\left(1 - \frac{F(u_1)}{ \sum_{i=1}^\ell F(u_i) + f(0)\left(1-\sum_{i=1}^\ell u_i\right) } \right)\;.\]
\end{lemma}

\begin{proof}
    This follows immediately from Corollary~\ref{C-KT}. If $f$ is constant in $[0, \alpha]$, then $F(x)=f(0)\,x$ for every $x \in [0, \alpha]$. Thus,
    \[\sum_{i=\ell+1}^k F(u_i) \EQ f(0)\sum_{i=\ell+1}^k u_i \EQ f(0)\left(1-\sum_{i=1}^\ell u_i\right)\;.\qedhere\]
\end{proof}

Note that the expression above depends only on the prefix $(u_1,u_2,\ldots,u_\ell)$ and not on the entire simplex point $(u_1,u_2,\ldots,u_k)$. 

\subsection{Further Analysis of IT}\label{sub:further-IT}

Bounding the density of IT when only a prefix $(u_1,u_2,\ldots,u_\ell)$ of a simplex point is given is harder. Karger et al.~\cite{KKSTY04} obtained such a bound when $\ell=2$ and when $f(x)$ is constant. This bound is not enough for our purposes and we need to extend it to larger values of~$\ell$. 

We show that if $F(x)$ is convex in the interval $[0, \alpha]$, where $\alpha > 0$, which is satisfied for example if~$f(x)$ is non-decreasing in $[0, \alpha]$, then the maximum in the definition of $d^{\IT(f),1}_k(u_1,u_2,\ldots,u_\ell)$, restricted to the region $0 \leq u_{\ell+1}, \ldots, u_k \leq \alpha$ (which always holds if $\sum_{i=1}^\ell u_i\ge 1- \alpha$),  is attained when $u_{\ell+1},\ldots,u_k$ are all equal. 

Before proceeding, we recall a special case of Maclaurin's inequality \cite{hardy1952inequalities} which is a refinement of the  inequality of arithmetic and geometric means.

\begin{lemma}[Maclaurin's inequality]\label{L-Maclaurin}
    If $x_1,x_2,\ldots,x_n\ge 0$ and $1\le k\le n$, then:
    \[ E_k(x_1,\ldots,x_n) \LE E_1(x_1,\ldots,x_n)^k \EQ
    \biggl(\frac{1}{n}\sum_{i=1}^n x_i\biggr)^k \;.\]
\end{lemma}

The following lemma now follows easily.

\begin{lemma}\label{L-sym-max}
    Let $F(x)$ be convex in $[0, \alpha]$. Then, the maximum of $E_k(x_1,\ldots,x_n)$, where $x_i=1-F(u_i)$, for $1\le i\le n$, under the constraints $0\le u_1,\ldots,u_n\le \alpha$ and $\sum_{i=1}^n u_i = s\le \alpha n$ is attained when $u_i=\frac{s}{n}$, for $i\in[n]$.
\end{lemma}

\begin{proof}
    Let $G(x)=1-F(x)$. As $F(x)$ is convex in $[0, \alpha]$, we get that $G(x)$ is concave in $[0, \alpha]$. Using Maclaurin's inequality (Lemma~\ref{L-Maclaurin}) and the concavity of~$G(x)$ we get:
    \[ E_k(x_1,\ldots,x_n) \LE \biggl(\frac{1}{n}\sum_{i=1}^n x_i\biggr)^k \EQ \biggl(\frac{1}{n}\sum_{i=1}^n G(u_i)\biggr)^k \LE G\biggl(\frac{\sum_{i=1}^n u_i}{n}\biggr)^k \EQ G\biggl(\frac{s}{n}\biggr)^k\;.\]
    The lemma follows as $G(\frac{s}{n})^k = E_k(1-F(\frac{s}{n}),\ldots,1-F(\frac{s}{n}))$.
\end{proof}

We are now ready to prove:

\begin{lemma}\label{L-IT-prefix}
    Let $F(x)$ be a convex function in $[0, \alpha]$, let $2\le \ell\le k$ and let $u_1,u_2,\ldots,u_\ell \ge 0$ be such that $1 - (k-\ell)\alpha \leq \sum_{i=1}^\ell u_i \leq 1$. Then, 
    
        \[
        \max_{\substack{0\leq u_{\ell + 1}, \ldots, u_{k} \leq \alpha \\ \sum_{i = 1}^ku_i = 1}}\,\,d^{\IT(f),1}_k(u_1,u_2,\ldots,u_\ell, u_{\ell + 1}, \ldots, u_k) \EQ d^{\IT(f),1}_k(u_1,u_2,\ldots,u_\ell,\underbrace{c,c,\ldots,c}_{k-\ell})\;,
        \]
    where $c=\frac{1-\sum_{i=1}^\ell u_i}{k-\ell}\le \alpha$.
    In particular, if $\sum_{i=1}^\ell u_i\ge 1- \alpha$, then
        \[ d^{\IT(f),1}_k(u_1,u_2,\ldots,u_\ell) \EQ d^{\IT(f),1}_k(u_1,u_2,\ldots,u_\ell,\underbrace{c,c,\ldots,c}_{k-\ell})\;.\]
\end{lemma}

\begin{proof}
    By Corollary~\ref{C-IT-decompose}, we get that $d^{\IT(f),1}_k(u_1,u_2,\ldots,u_k)$ is a weighted sum, with nonnegative coefficients, of terms of the form $e_b(y_{\ell+1},\ldots,y_k)$, where $0\le b\le k-1$ and $y_i=1-F(u_i)$, for $\ell+1\le i\le k$. By Lemma~\ref{L-sym-max}, under the condition that $0\leq u_{\ell + 1}, \ldots, u_{k} \leq \alpha$, the maxima of $E_b(y_{\ell+1},\ldots,y_k)$, and hence also of $e_b(y_{\ell+1},\ldots,y_k)$, for all values of~$b$, are attained simultaneously when $u_{\ell+1}=\cdots=u_n=c=(1-\sum_{i=1}^\ell u_i)/(k-\ell)$. Hence, these are also the values of $u_{\ell+1},\ldots,u_k$ that maximize $d^{\IT(f),1}_k(u_1,u_2,\ldots,u_k)$.
\end{proof}

We next want to determine $d^{\IT(f),1}_\infty(u_1,u_1,\ldots,u_\ell)$, as in Definition~\ref{D:inf}. The expression we obtain involves Kummer's confluent hypergeometric function 
\[{}_1F_1(a;b;z) \EQ \sum_{n=0}^\infty \frac{(a)_n}{(b)_n}\frac{z^n}{n!}\;,\] also called the confluent hypergeometric function of the first kind, where $(a)_n=a(a+1)\ldots(a+n-1)$ is the rising Pochhammer symbol. (Note that $(a)_0=1$.) 

\begin{lemma}\label{L-IT-asymptotic}
    Suppose that $F(x)$ is convex in $[0,\alpha]$, let $u_1,u_2,\ldots,u_\ell \ge 0$ be such that $\sum_{i=1}^\ell u_i \leq 1$. Let $\beta=f(0)(1-\sum_{i=1}^\ell u_i)$. As before, let $y_i=1-F(u_i)$, for $i\in[2,\ell]$. Then
    \[
        \lim_{k \to \infty}d^{\IT(f),1}_k(u_1,u_2,\ldots,u_\ell,\underbrace{c,c,\ldots,c}_{k-\ell}) \EQ f(u_1) \cdot \sum_{a=0}^{\ell-1} c_{\ell,a}(\beta)\, e_a(y_2,\ldots,y_\ell)\;,\\
    \]
    where  $c=(1-\sum_{i=1}^\ell u_i)/(k-\ell)$ inside the limit expression and 
    \begin{equation}\label{eq:c_ell_a}
    c_{\ell,a}(\beta) \EQ \int_0^1 t^a (1-t)^{\ell-1-a}{\rm e}^{-\beta t}\,dt \EQ \frac{1}{(\ell-a)\binom{\ell}{a}}\cdot{}_1F_1(a+1;\ell+1;-\beta)\;.
    \end{equation}
\end{lemma}

\begin{proof}
    By Corollary~\ref{C-IT-decompose} we have:
    \[ d^{\IT(f),1}_k(u_1,u_2,\ldots,u_k) \EQ f(u_1) \cdot \frac{1}{k} \sum_{\substack{0\le a<\ell\\0\le b\le k-\ell\\a+b\le k-2}} \frac{1}{\binom{k-1}{a+b}} e_a(y_2,\ldots,y_\ell)\cdot e_b(y_{\ell+1},\ldots,y_k) \;.\]
    As it makes no asymptotic difference, we ignore the condition $a+b\le k-2$. (Note that the conditions on~$a$ and $b$ imply that $a+b\le k-1$.) 
    It follows that 
    \begin{align*}
        &\,d^{\IT(f),1}_k(u_1,u_2,\ldots,u_\ell,\underbrace{c,c,\ldots,c}_{k-\ell}\,) \\
        \EQ&\,f(u_1) \cdot \frac{1}{k}\sum_{\substack{0\le a<\ell\\0\le b\le k-\ell}} \frac{1}{\binom{k-1}{a+b}} e_a(y_2,\ldots,y_\ell)\, \cdot \binom{k-\ell}{b}(1-F(c))^b\\
        \EQ&\,f(u_1) \cdot \sum_{a=0}^{\ell-1}\underbrace{\left[\frac{1}{k}\sum_{b=0}^{k-\ell} \frac{\binom{k-\ell}{b}}{\binom{k-1}{a+b}}(1-F(c))^b\right]}_{c_{\ell,a}(k,c)} e_a(y_2,\ldots,y_\ell)\;.
    \end{align*}
    As $c=\frac{1-\sum_{i=1}^\ell u_i}{k-\ell}$, we have $F(c)\approx f(0)\frac{1-\sum_{i=1}^\ell u_i}{k-\ell}\approx \frac{\beta}{k}$ as $k\to\infty$. We need to determine the limits of the coefficients
    \[ c_{\ell,a}(\beta) \EQ \lim_{k\to\infty} c_{\ell,a}(k,c)
    \EQ \lim_{k\to\infty} \frac{1}{k}\sum_{b=0}^{k-\ell} \frac{\binom{k-\ell}{b}}{\binom{k-1}{a+b}}\left(1-\frac{\beta}{k}\right)^b\;.\]
    It is known that
    \[ \frac{1}{\binom{n}{k}} \EQ (n+1)\int_0^1 t^k(1-t)^{n-k}\,dt\;.\]
    Thus,
    \[ \frac{1}{k}\sum_{b=0}^{k-\ell} \frac{\binom{k-\ell}{b}}{\binom{k-1}{a+b}}\left(1-\frac{\beta}{k}\right)^b
    \EQ \sum_{b=0}^{k-\ell} \binom{k-\ell}{b} \left(1-\frac{\beta}{k}\right)^b \int_0^1 t^{a+b}(1-t)^{k-1-a-b}\,dt\]
    \[ \EQ \int_0^1 t^a (1-t)^{k-1-a} \sum_{b=0}^{k-\ell} \binom{k-\ell}{b} \left(\frac{t(1-\frac{\beta}{k})}{1-t}\right)^b\,dt 
    \EQ \int_0^1 t^a (1-t)^{k-1-a} \left(1+\frac{t(1-\frac{\beta}{k})}{1-t}\right)^{k-\ell}dt\]
    \[ \EQ \int_0^1 t^a (1-t)^{k-1-a} \left(\frac{1-\frac{\beta t}{k}}{1-t}\right)^{k-\ell}dt
    \EQ \int_0^1 t^a (1-t)^{\ell-1-a} \left(1-\frac{\beta t}{k}\right)^{k-\ell}dt \;.\]
    As $k\to\infty$, the above expression is asymptotically
    \[ \int_0^1 t^a (1-t)^{\ell-1-a} \ee^{-\beta t}\,dt \EQ B(a+1,\ell-a)\cdot{}_1F_1(a+1;\ell+1;-\beta) \EQ \frac{1}{(\ell-a)\binom{\ell}{a}}\cdot{}_1F_1(a+1;\ell+1;-\beta)\;, \]
    where $B(a+1,\ell-a)=\frac{\Gamma(a+1)\Gamma(\ell-a)}{\Gamma(\ell+1)}=\frac{a! (\ell-a-1)!}{\ell!}$ is the beta function. The next to last equation is a known identity involving the function ${}_1F_1(a+1;\ell+1;-\beta)$, see \cite[Eq.~(13.4.1)]{DLMF-13}.
\end{proof}

\begin{corollary}
Let $F(x)$ be a convex function in $[0, \alpha]$, let $2\le \ell\le k$ and let $u_1,u_2,\ldots,u_\ell \ge 0$ be such that $1 - (k-\ell)\alpha \leq \sum_{i=1}^\ell u_i \leq 1$. Then, 
    
        \[
        \max_{\substack{0\leq u_{\ell + 1}, \ldots, u_{k} \leq \alpha \\ \sum_{i = 1}^ku_i = 1}}\,\,d^{\IT(f),1}_k(u_1,u_2,\ldots, u_k) \LE  f(u_1) \cdot \sum_{a=0}^{\ell-1} c_{\ell,a}(\beta)\, e_a(y_2,\ldots,y_\ell)\;,\\
        \]
    where $c_{\ell, a}(\beta)$ is defined as in \eqref{eq:c_ell_a}. In particular, if $\sum_{i=1}^\ell u_i \geq 1 - \alpha$, then
    \[ 
    \displaystyle
    d^{\IT(f),1}_\infty(u_1,u_2,\ldots,u_\ell) \EQ f(u_1) \cdot \sum_{a=0}^{\ell-1} c_{\ell,a}(\beta)\, e_a(y_2,\ldots,y_\ell)\;.
    \]   
\end{corollary}
\begin{proof}
    By Lemma~\ref{L-IT-prefix}, we have for $c=(1-\sum_{i=1}^\ell u_i)/(k-\ell)$,
        \[
        \max_{\substack{0\leq u_{\ell + 1}, \ldots, u_{k} \leq \alpha \\ \sum_{i = 1}^ku_i = 1}}\,\,d^{\IT(f),1}_k(u_1,u_2,\ldots, u_k) \EQ d^{\IT(f),1}_k(u_1,u_2,\ldots,u_\ell,\underbrace{c,c,\ldots,c}_{k-\ell})\;,
        \]
    In particular, the expression on the right hand side is non-decreasing as $k$ grows. It then follows from Lemma~\ref{L-IT-asymptotic} that
    \begin{align*}
        & \, \max_{\substack{0\leq u_{\ell + 1}, \ldots, u_{k} \leq \alpha \\ \sum_{i = 1}^ku_i = 1}}\,\,d^{\IT(f),1}_k(u_1,u_2,\ldots,u_\ell, u_{\ell + 1}, \ldots, u_k) \\
        \LE &\, \lim_{k \to \infty}d^{\IT(f),1}_k(u_1,u_2,\ldots,u_\ell,\underbrace{c,c,\ldots,c}_{k-\ell}) \\
        \EQ &\, f(u_1) \cdot \sum_{a=0}^{\ell-1} c_{\ell,a}(\beta)\, e_a(y_2,\ldots,y_\ell)\;.
    \end{align*}
    
    If $\sum_{i=1}^\ell u_i \geq 1 - \alpha$, then any new coordinate after $u_\ell$ must be at most $\alpha$, so
    \[
    d^{\IT(f),1}_\infty(u_1,u_2,\ldots,u_\ell) \EQ \lim_{k \to \infty} d^{\IT(f),1}_k(u_1,u_2,\ldots,u_\ell) \EQ f(u_1) \cdot \sum_{a=0}^{\ell-1} c_{\ell,a}(\beta)\, e_a(y_2,\ldots,y_\ell)\;.\qedhere
    \]
\end{proof}

Regarding the coefficients $c_{\ell, a}$, it is interesting to note that $c_{\ell,a}(\beta)=\E[\ee^{-\beta T}]$, where $T\sim \text{Beta}(a+1,\ell-a)$ is a beta random variable. We have no intuitive explanation of this identity. As special cases we get:
\[
c_{2,0}(\beta) \EQ \int_0^1 (1-t)\ee^{-\beta t}\,dt \EQ \frac{\ee^{-\beta}-1+\beta}{\beta^{2}}
\quad,\quad
c_{2,1}(\beta) \EQ \int_0^1 t\,\ee^{-\beta t}\,dt \EQ \frac{1-(1+\beta)\ee^{-\beta}}{\beta^{2}} \;.
\]
This matches the results of \cite{KKSTY04} for $\ell=2$.

By manipulating the expression obtained in Lemma~\ref{L-IT-asymptotic} it is possible to also obtain the following finite formula that does not explicitly involve the hypergeometric function:
\[ c_{\ell,a}(\beta) \EQ \sum_{j=0}^{\ell-a-1}(-1)^j\binom{\ell-a-1}{j}\,
\frac{(a+j)!}{\beta^{\,a+j+1}}\,
\Bigl(1-e^{-\beta}\sum_{r=0}^{a+j}\frac{\beta^r}{r!}\Bigr) \;.\]

\section{Discovery of Combined Rounding Schemes}\label{sec:minimax}

Now that we have discussed in detail the properties of our basic schemes, we turn to discuss how they can be combined to produce a mixture which beats the state of the art for \MC. The purpose of this section is to describe the high-level details of the computational procedure we used to discover the rounding schemes needed to prove Theorem~\ref{thm:main} and Theorem~\ref{thm:finite}. The rounding schemes found are described in Section~\ref{sec:new-schemes}.

\subsection{Zero-sum Games and the Minimax Theorem}

The main inspiration for our discovery procedure is a recent algorithm developed by Brakensiek, Huang, Potechin, and Zwick~\cite{BHPZ23} for discovering new rounding schemes for MAX DI-CUT and for related problems. However, the exact settings are quite different. In particular, the MAX DI-CUT rounding schemes round solutions of an SDP relaxation, while here we round solutions of an LP relaxation. What the two problems have in common is that they can both be thought of as $2$-player $0$-sum games.

In the context of \MC, we think of the two players as a \emph{cut player} and an \emph{edge player}. To make the game finite, assume the cut player has a finite list $\bR = \{R_1, \hdots, R_n\}$ of rounding schemes for \MC\ (in our case these will be the aforementioned \emph{basic} rounding schemes), and that the edge player has a finite set $\bU = \{\bu_1, \hdots, \bu_m\} \subset \Delta_k$ of simplex points.\footnote{We call the player the edge player and not the vertex player as the simplex points represent infinitesimal $(1,2)$-aligned edges.} We assume for now that~$k$, the number of terminals, is fixed. (For a discussion of the underlying infinite game, see Section~\ref{sec:discussion}.) The game proceeds with the cut player picking $i \in [n]$ and the edge player picking $j \in [m]$ simultaneously. The payoff for the edge player is then $d^{R_i}_k(\bu_j)$ and the payoff for the cut player is $-d^{R_i}_k(\bu_j)$. In other words, the cut player seeks to \emph{minimize} $d^{R_i}_k(\bu_j)$ while the edge player seeks to \emph{maximize} $d^{R_i}_k(\bu_j)$.

From the classical theory of $0$-sum games~\cite{neumann1928theorie}, we know that the game has a \emph{value} and a \emph{Nash equilibrium}~\cite{nash1951non} of mixed strategies. That is, there exists probability distribution $(p_1, \hdots, p_n)$ and $(q_1, \hdots, q_m)$ such that
\begin{align}
    \sum_{i=1}^n \sum_{j=1}^m p_iq_j\, d^{R_i}_k(\bu_j) \EQ \min_{i \in [n]} \sum_{j=1}^m q_j\, d^{R_i}_k(\bu_j) \EQ \max_{j \in [m]} \sum_{i=1}^n p_i\, d^{R_i}_k(\bu_j)\;.\label{eq:minimax}
\end{align}
The common value above is called the \emph{value} of the zero-sum game, which we denote by $\val(\bR, \bU)$. The distributions $(p_1, \hdots, p_n)$ and $(q_1, \hdots, q_m)$ are optimal mixed strategies of the two players. The value and the optimal strategies can be found efficiently by solving a linear program.

In general, $\val(\bR,\bU)$ gives neither an upper nor a lower bound on the approximation ratio of the \MC\ problem, since $\bR$ and $\bU$ may not include useful rounding schemes and simplex points.  We could, in principle, get an approximate upper bound on the ratio by choosing $\bU$ to be an $\eps$-net of the simplex. This, however, is impractical, since the size of such an $\eps$-net is $O(\eps^{1-k})$. Our approach, following \cite{BHPZ23}, is to adaptively grow~$\bR$ and $\bU$ until the value of the game cannot be changed by much by adding new rounding schemes or simplex points.

The discovery process is used to produce \emph{candidate} mixtures of rounding schemes that hopefully have good approximation ratios. These mixtures are subjected to a rigorous verification as described in Section~\ref{sec:verification}. There are thus no correctness concerns when it comes to the discovery process.

\subsection{Finding New Basic Rounding Schemes}

Let $\bR := \{R_1, \hdots, R_n\}$ and $\bU := \{\bu_1, \hdots, \bu_m\}$ be finite collections of basic rounding schemes and simplex points. Let $(p_1, \hdots, p_n)$ and $(q_1, \hdots, q_m)$ be optimal distributions of the two players, and let~$v$ be the value of the game. It is clear that using only a mixture of $\bR$ we cannot obtain an approximation ratio better than~$v$. We would thus like to find a new basic rounding scheme $R'$ whose addition to $\bR$ will enable us to obtain a better approximation ratio. Ideally, we would like to find a rounding scheme $R'$ that minimizes $\sum_{j=1}^m q_j\,d^R_k(\bu_j)$. 

Even if we restrict ourselves to finding the best rounding scheme $R'$ from one of the studied families of rounding schemes, this seems to be a daunting task. Suppose, for example, that we want to find the best GKT scheme. We then seek to find a density function $f$ for which $\sum_{j=1}^m q_j\,d^{\KT(f)}_k(\bu_j)$ is minimized.
This is a \emph{calculus of variations} problem as we are optimizing over functions rather than points in $\RR^n$ for some finite~$n$. Solving such problems, even approximately, seems out of reach at present.

To make the optimization problem more manageable, we need to restrict the search to density functions from fairly restricted families of functions. We chose to work with the family of \emph{piecewise linear} density functions composed of a limited number of pieces. A function $f:[0,1]\to\RR^+$ is piecewise linear if there exists $0=x_0<x_1<\dots x_r=1$ such that $f(x)$ is linear in each interval $[x_{i-1},x_i]$, for $i\in [r]$. If we let $y_i=f(x_i)$, for $i\in[0,r]$, then the $k+1$ pairs $(x_0,y_0),(x_1,y_1),\ldots,(x_r,y_r)$ determine the function exactly. For technical reasons, we also assume that $x_i-x_{i-1}\ge \delta$, for $i\in[r]$, for some positive parameter $\delta>0$. As~$r$ increases, and $\delta$ decreases, this provides better and better approximations of all continuous density functions. We are, however, restricted to work with fairly small values of~$r$.

Finding the piecewise linear function~$f$ with at most $r$ pieces that minimizes $\sum_{j=1}^m q_j\,d^{\KT(f)}_k(\bu_j)$ is still a hard non-convex optimization problem. That said, we are now optimizing over a finite number of variables and can thus use standard numerical techniques to try to find at least a local minimum.

An important aspect of our approach is that we are not dependent on finding the best density function $f$. Rather, we are making progress as long as we manage to find some density function~$f$ for which $\sum_{j=1}^m q_j\,d^{\KT(f)}_k(\bu_j)$ is smaller than the current value of the game. To find such a function we perform many numerical searches starting from various initial points, including, in particular, random initial points. 

To speed-up the verification process, the schemes to be verified should satisfy a condition that we call $\alpha$-compliance with $\alpha$ as large as possible. (For the details, see Section~\ref{sec:verification}.)

\begin{definition}[$\alpha$-compliance]\label{D:compliance}
    Let $\alpha \in (0, 1]$. We say that a mixture $\mathcal{R}$ over GKT, IT, ST, and DT rounding schemes is $\alpha$-compliant, if the following holds:
            
        \begin{itemize}
            \item For every $\IT(f)$ in the support of $\mathcal{R}$, $f$ is non-decreasing on $[0, \alpha]$;
            \item For every $\KT(f)$ in the support of $\mathcal{R}$, $f$ is constant and non-zero on $[0, \alpha]$.
        \end{itemize}
\end{definition}

It is of course easy to restrict the search new rounding functions so that the resulting schemes would be $\alpha$-compliant. In our schemes we typically have $\alpha=0.25$ or $\alpha=0.15$. 

\subsection{Finding New Simplex Points}

In our rounding scheme discovery, we do not want to \emph{only} find new basic rounding schemes,  as our finite sample of simplex points $\bU$ may be insufficient to accurately approximate the approximation ratio of the optimal mixture of schemes in $\bR$. To mitigate this, we optimize over potential simplex points $\bu \in \Delta_k$ to add to $\bU$.

More formally, we seek to solve the following optimization problem: $\max_{\bu \in \Delta_k} \sum_{i=1}^n p_i\,d^{R_i}_k(\bu)$. Unlike the optimization to find new basic rounding schemes, this is a rather standard non-linear optimization problem, so off-the-shelf methods suffice to find local maxima $\bu$.

However, unlike the task of finding new rounding schemes, to eventually certify the correctness of our rounding scheme (see Section~\ref{sec:verification}), we need much more assurance that there is not a significant global optimum we are missing. Toward this, we use a few strategies. First, we run the local optimizer from many more random points than we do for finding random schemes. Second, we also use all existing points in $\bU$ as starting points. The reason for this is that if we have been running many rounds of optimization already, the point in $\bU$ are likely close to (near) global optima, but the optima may have shifted slightly due to the addition of new schemes in $\bR$.

Another important strategy is to not search blindly over all $\bu \in \Delta_k$, but rather the inequalities in Section~\ref{sec:further} suggest that the optimal $\bu \in \Delta_k$ often have many coordinates of $\bu$ equal. By imposing constraints that some of the coordinates of $\bu$ must be equal, we can greatly reduce the dimension of our optimization, making it more likely we find global optima.

The discussion so far assumed that~$k$ is fixed, and rather small. Additional complications arise when~$k$ is arbitrary. In this case, we work with prefixes of simplex points, as explained in Section~\ref{sec:further}. More precisely, we use $\bu=(u_1,u_2,\dots,u_\ell)$, for $\ell\ge 2$, to represent a point $\bu=(u_1,u_2,\dots,u_\ell,c,c,\dots,c)$, where $c=(1-\sum_{i=1}^\ell u_i)/(k-\ell)$ and $k\to\infty$. This prefixed-based approach allows us to compactly represent the most difficult simplex points for our rounding schemes.

\subsection{The Heuristic Discovery Algorithm}

\newcommand{\AlgoComment}[2][1]{%
  \Statex\hspace*{#1\algorithmicindent}~#2%
}

\newlength{\algtmp}
\newcommand{\AlgTightSkip}[1][.3ex]{%
  \setlength{\algtmp}{\itemsep}%
  \Statex\vspace{-\algtmp}\vspace{#1}%
}

\begin{algorithm}[t]
    \caption{The Discovery algorithm for a fixed $k$}\label{alg:discover}
    \begin{algorithmic}[1]
        \AlgTightSkip[5pt] 
        \State \textbf{Input:} Initial collection of basic rounding schemes~$\bR_0$, an initial collection of simplex points~$\bU_0$, initial and final precision parameters $\eps_0$ and $\eps_1$, and a scale factor $\gamma$.
        \AlgTightSkip[-10pt] 
        \Procedure{Discover}{$\bR_0,\bU_0,\eps_0,\eps_1,\gamma$}
            \AlgTightSkip[-10pt] 
            \State $\bR,\bU\gets \bR_0,\bU_0$
            \State $v,p,q\gets \textsc{SolveGame}(\bR,\bU)$
            \State $\eps\gets \eps_0$
            \AlgTightSkip[-10pt] 
            \While {$\eps\ge\eps_1$}
                \AlgoComment[2]{\(\triangleright\) Discover new basic rounding schemes}
                \For {$\RT\in\{\ST,\KT,\IT,\DT\,\}$}
                    \For {$i\gets 1$ to $N$}
                        \State $\flag\gets \textbf{false}$
                        \For {$j\gets 1$ to $M$}
                            \State $f,v'\gets \textsc{NewScheme}(\RT,\bU,q)$
                            \If {$v'\le v-\eps$}
                                \State $\bR\gets \bR\cup\{RT(f)\}$
                                \State $v,p,q\gets\textsc{SolveGame}(\bR,\bU)$
                                \State $\flag\gets \textbf{true}$
                                \State \textbf{break}
                            \EndIf
                        \EndFor
                        \If {\textbf{not} $\flag$}
                            \State \textbf{break}
                        \EndIf
                    \EndFor
                \EndFor
                \AlgoComment[2]{\(\triangleright\) Discover new simplex points}                    
                \For {$i\gets 1$ to $N$}
                    \State $\flag\gets \textbf{false}$
                    \For {$j\gets 1$ to $M$}
                        \State $\bu,v'\gets \textsc{NewPoint}(\bR,p)$
                        \If {$v'\ge v+\eps$}
                            \State $\bU\gets \bU\cup\{\bu\}$
                            \State $v,p,q\gets\textsc{SolveGame}(\bR,\bU)$
                            \State $\flag\gets \textbf{true}$
                            \State \textbf{break}
                        \EndIf
                    \EndFor
                    \If {\textbf{not} $\flag$}
                        \State \textbf{break}
                    \EndIf
                \EndFor
                \AlgTightSkip[-10pt] 
                \State $\eps\gets \gamma \eps$
            \EndWhile
            \State \Return $v,(\bR,p),(\bU,q)$
        \EndProcedure
    \end{algorithmic}

\end{algorithm}

We now put all the pieces together and describe the algorithm used to discover our improved rounding schemes. We assume that we have the following three basic functions:
\begin{itemize}
    \item[] \textsc{SolveGame}$(\bR, \bU)$ -- Solves the game defined by $\bR$ and $\bU$ by solving a linear program. Returns a triplet $v,p,q$, where $v$ is the value of the game and $p$ and $q$ are the optimal strategies of the two players.
    \item[] \textsc{NewScheme}$(\RT,\bU,q)$ -- Tries to find a piecewise linear density function~$f$ composed of at most~$r$ pieces that minimizes $\sum_{j=1}^m q_j\,d^{\RT(f)}_k(\bu_j)$, where $\RT\in\{\ST,\KT,\IT,\DT\,\}$. Returns a pair $f,v$, where $f$ is the function found and $v$ is its value.
    \item[] \textsc{NewPoint}$(\bR,p)$ -- Tries to find a simplex point $\bu$ that maximizes $\sum_{i=1}^n p_i\,d^{R_i}_k(\bu)$. Returns a pair $\bu,v$, where $\bu$ is the point found and $v$ is its value.
\end{itemize}

The discovery algorithm \textsc{Discover}$(\bR_0,\bU_0,\eps_0,\eps_1,\gamma)$, whose pseudocode is given in Algorithm~\ref{alg:discover}, receives an initial collection $\bR_0$ of basic rounding schemes, an initial collection $\bU_0$ of simplex points, initial and final precision parameters $\eps_0$ and $\eps_1$, and a scale factor~$\gamma$. The algorithm also receives many other parameters that for simplicity are not shown explicitly. Three of these parameters are~$\alpha$ - the compliance parameter, $r$ - the maximum number of pieces that the piecewise linear density functions may have, and~$\delta$ - the minimum length of a piece.

The algorithm works in rounds. In each round it first tries to find new basic rounding schemes. More precisely, it tries to add up to~$N$ rounding schemes from each one of the four basic families. To find each such scheme it makes at most~$M$ calls to \textsc{NewScheme}, where $N$ and $M$ are some parameters. A new scheme is added only if it improves the value by at least~$\eps$. Next, the algorithm tries to add up to~$N$ new simplex points. To find each such simplex points it makes at most~$M$ calls to \textsc{NewPoint}. \textsc{Discover} also tries to add simplex points by starting searches from existing simplex points. This is not shown explicitly in the pseudocode. At the end of each round, $\eps$ is scalled down by a factor of~$\gamma$. When $\eps<\eps_1$, the algorithm returns $v$, $(\bR,p)$ and $(\bU,q)$. If $\eps_1$ is small enough, and~$N$ and~$M$ are large enough, then $\cR=\sum_{i=1}^n p_i R_i$ is a rounding scheme likely to give an approximation ratio not much larger than $v+\eps_1$. This is then verified rigorously. 

\subsection{Implementation Details}

Our implementation of \textsc{Discover} was done in Python. \textsc{SolveGame} solves the linear programs using Gurobi \cite{gurobi}, with the help of an API wrapper provided by CVXPY~\cite{diamond2016cvxpy}. The heuristics for finding new distributions \textsc{NewScheme} and \textsc{NewPoint} were implemented using the \textsc{scipy.optimize} library~\cite{2020SciPy-NMeth}. Of the heuristic provided by SciPy, ``Sequential Least Squares Programming (SLSQP)"~\cite{kraft1988software,lawson1995solving,nocedal2006numerical} was the most effective for solving our non-linear optimization problems. We also note that more than 99\% of the running time was spent looking for new rounding schemes and new simplex points. Less than 1\% of the running time was spent on solving linear programs.

Finding each new rounding scheme involved several runs of \textsc{Discover}, each with some adjustment of the parameters. Each one of these runs could of course start with the rounding scheme produced by the previous run. Most of these runs took more than a day.

The current implementation of \textsc{Discover} is sequential, as we resolve the game whenever a new basic rounding scheme or a new simplex point is added. We may try to experiment with a parallel implementation of \textsc{Discover} in which the game is resolved only after finding the $N$ new basic rounding schemes, or the $N$ new simplex points in each round.

\section{The New Rounding Schemes}\label{sec:new-schemes}

\begin{figure}[t]
    \centering
    \includegraphics[scale=0.5]{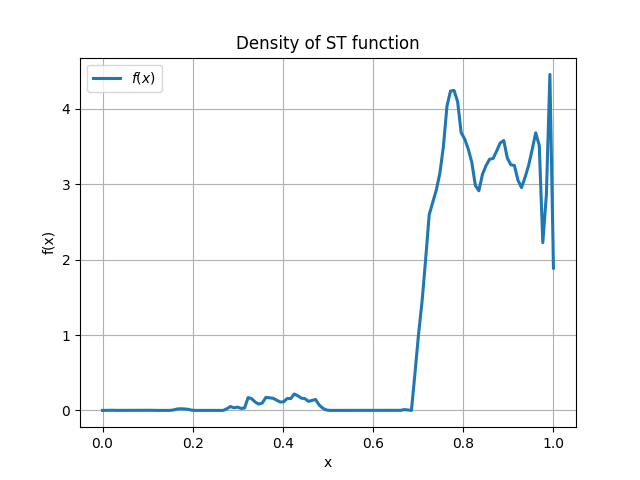}
    \includegraphics[scale=0.5]{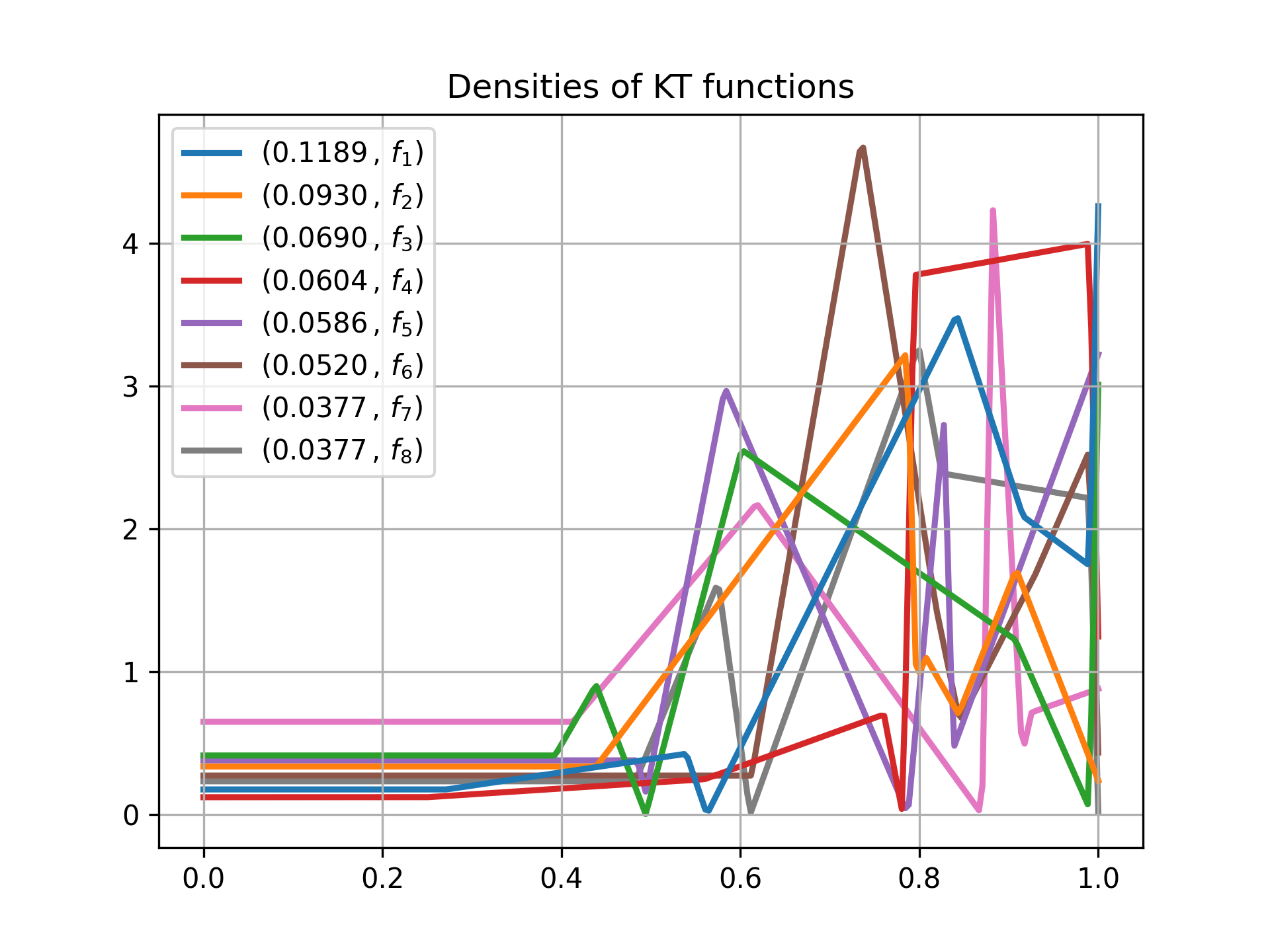}
    \includegraphics[scale=0.5]{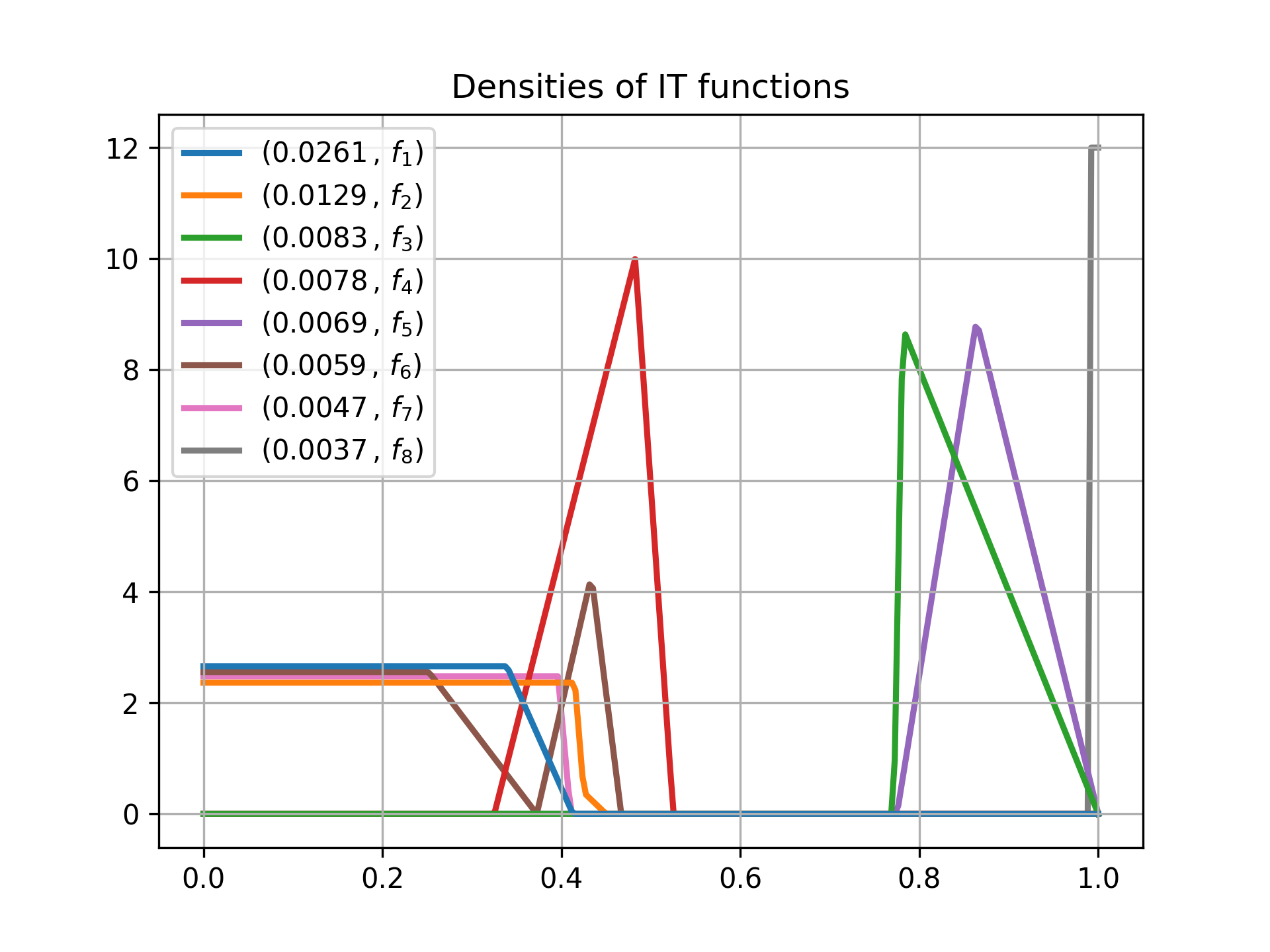}
    \includegraphics[scale=0.5]{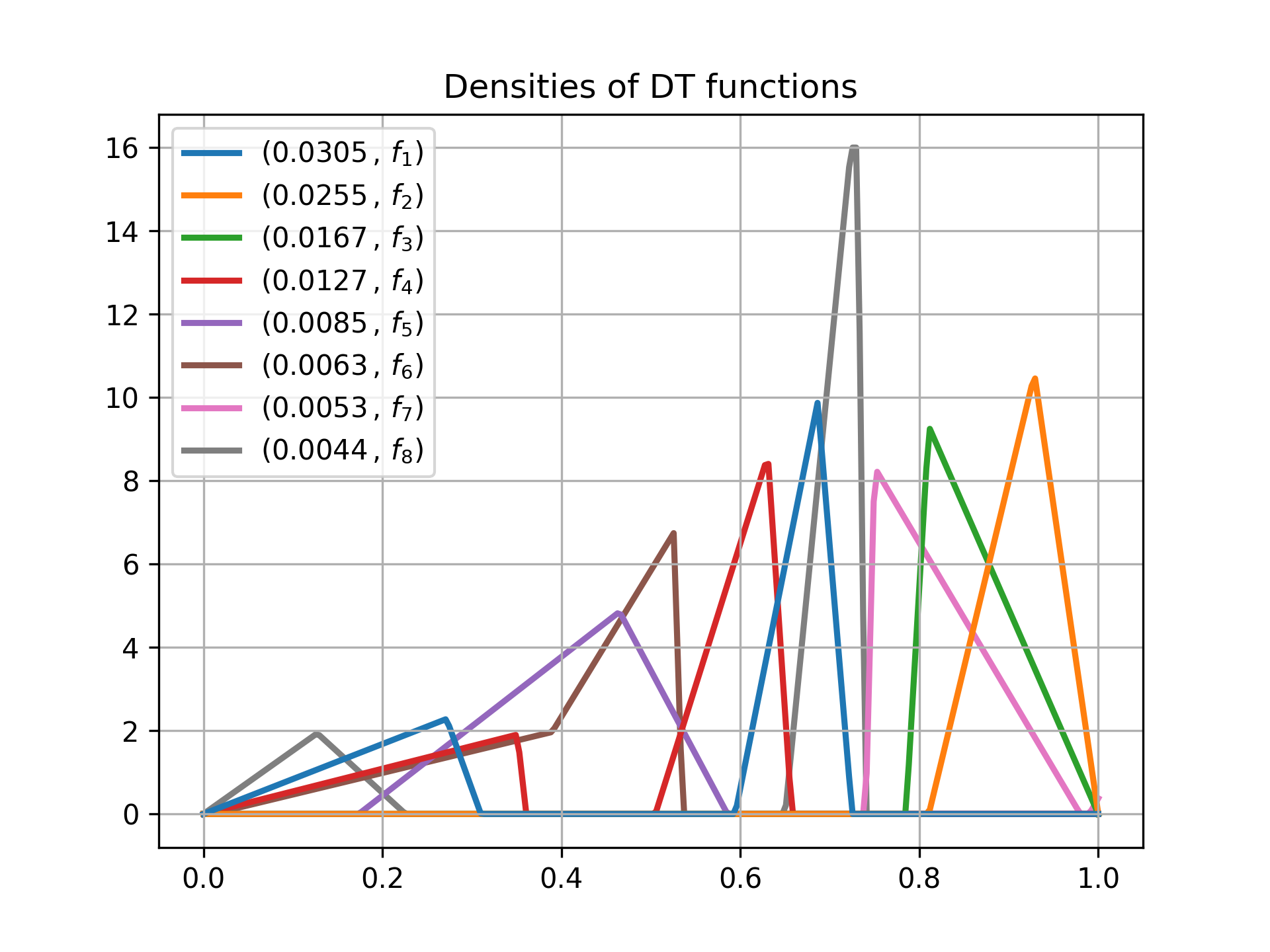}
    \caption{The density functions of the ST, KT, IT and DT schemes used in mixture that obtains an approximation ratio of~$\ourAPX$ for an arbitrary number of terminals. For ST, the combined density of all ST schemes is given. For each one of KT, IT and DT, the densities of the 8 rounding schemes of each family used with the largest probabilities are shown. Also shown are the probabilities with which each one of these rounding schemes is used. Plots in this paper are made with Matplotlib~\cite{Hunter:2007}.}
    \label{F-densities}
\end{figure}

\begin{figure}[t]
    \centering
    \includegraphics[scale=0.4]{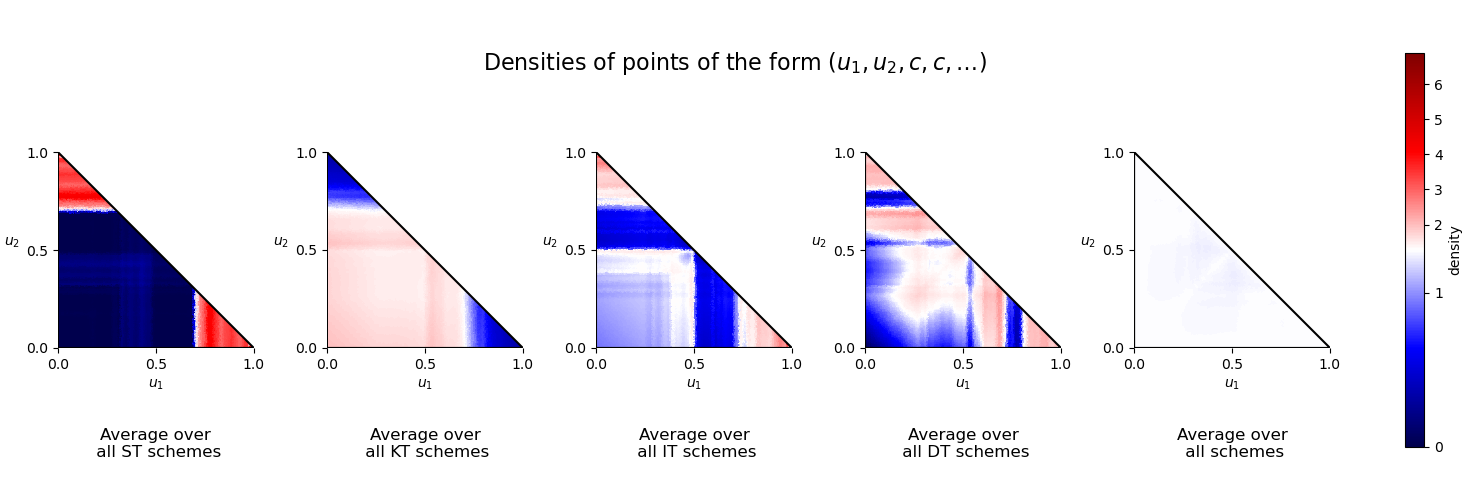}
    \includegraphics[scale=0.4]{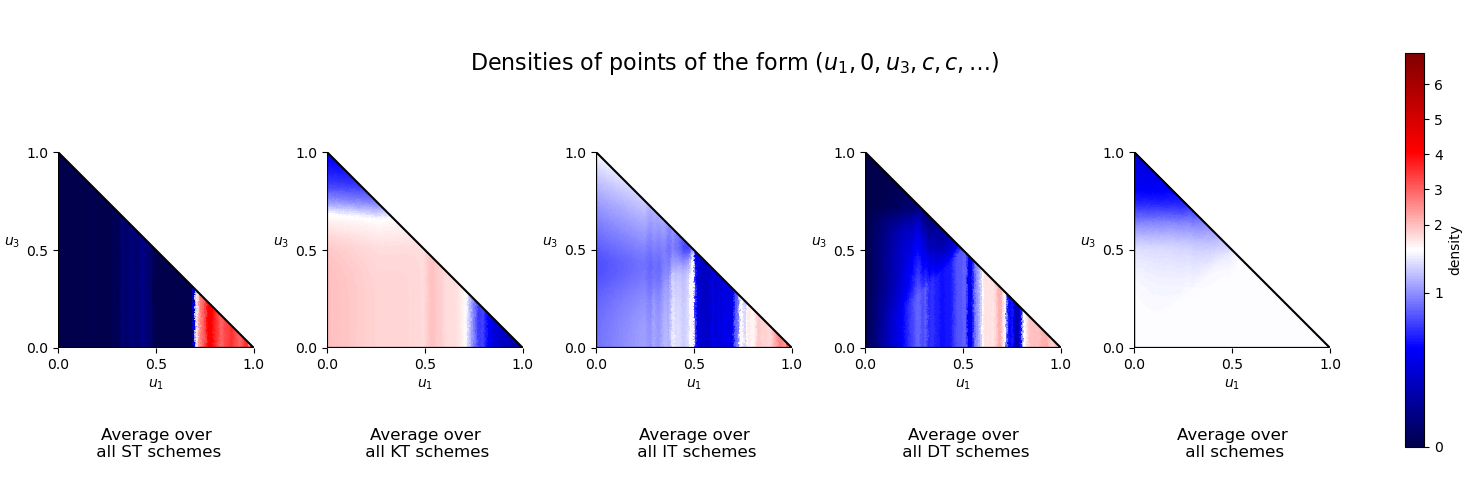}
    \includegraphics[scale=0.4]{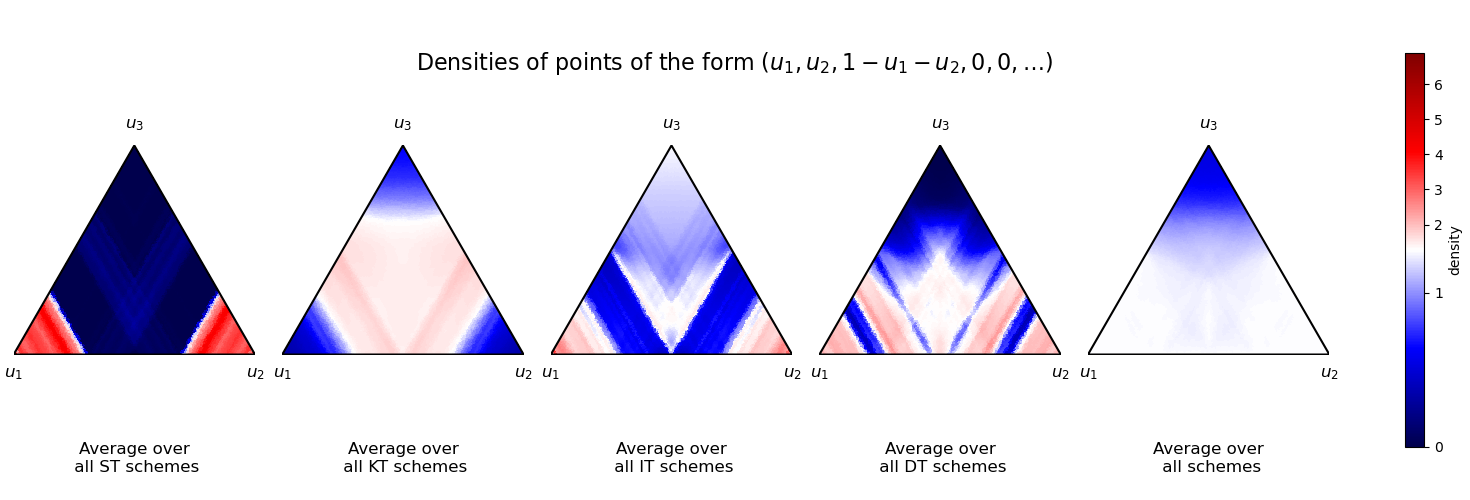}
    \caption{Heatmaps showing the densities of the various components of the algorithm, and of the whole algorithm, on simplex points of the form $(u_1,u_2,c,c,\ldots)$, $(u_1,0,u_3,c,c,\ldots)$ and $(u_1,u_2,u_3,0,\ldots,0)$. In points of the first form we have $c=(1-u_1-u_2)/(k-2)$, where $k\to\infty$, and similarly for points of the second form. The boundary value between blue and red shades is $1.278$. Although the different components have very large maximum densities, the mixture of all of them gives a final scheme in which the densities are almost constant in most the regions shown. Plots in this paper are made with Matplotlib~\cite{Hunter:2007}.}
    \label{F-heatmaps}
\end{figure}

\begin{figure}[t]
    \centering
    \includegraphics[scale=0.35]{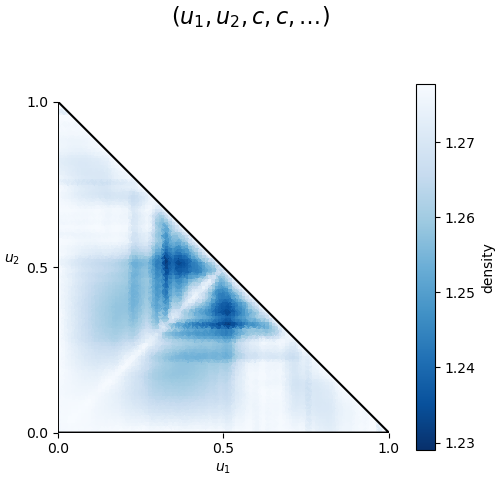}
    \includegraphics[scale=0.35]{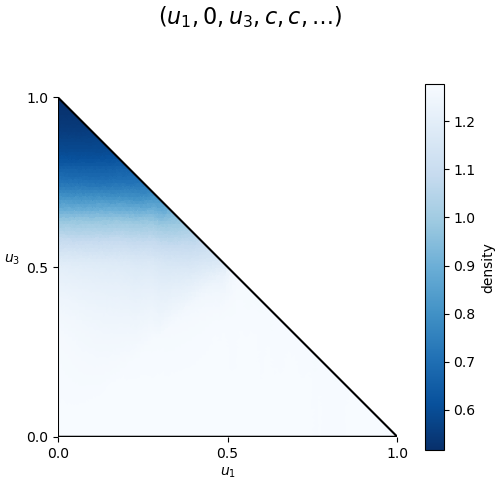}
    \includegraphics[scale=0.35]{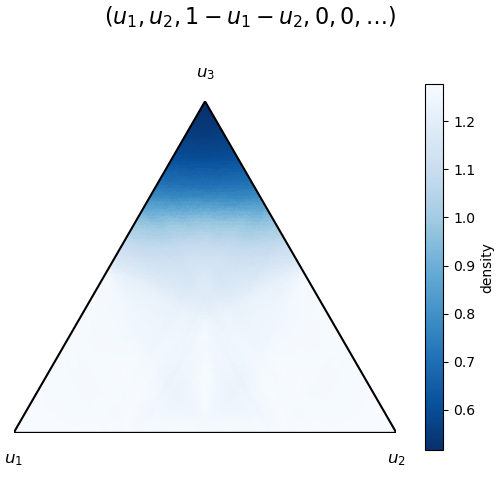}
    \caption{A closer look at the densities of the whole algorithm.}
    \label{F-closeup}
\end{figure}

We now give more details about the improved rounding schemes found by our computational approach. We concentrate on the scheme found for an arbitrary number of terminals. The details for the schemes found for small number of terminals are similar.

The best scheme we found so far for arbitrary~$k$ is a mixture of \numfun\ basic rounding schemes. Each basic rounding scheme is an ST$(f)$, KT$(f)$, IT$(f)$ or DT$(f)$ scheme, where $f$ is its piecewise linear density function, and used with a specific probability. (Due to a technical reason we allow several ST$(f)$ schemes. They could be replaced, however, by one equivalent ST$(f')$ scheme.) The scheme was chosen to be $0.25$-compliant (see Definition~\ref{D:compliance}). The \numfuna\ schemes with the largest probabilities are used with a total probability of $0.99$. KT, DT, ST and IT schemes are used with total probabilities of approximately $\KTprob$, $\DTprob$, $\STprob$ and $\ITprob$, respectively. It is interesting to note that the newly introduced family of GKT schemes is used with a probability of almost $0.6$, showing its great usefulness. Each of the other three families is used with a probability of about $0.14$, showing that each one of them is still important.

Figure~\ref{F-densities} shows the density functions used by some of the basic rounding schemes of the mixture. For ST, the combined density function of all ST schemes is shown. Also shown are the density functions of the 8 basic rounding schemes from each of the families KT, IT and DT that are used with the largest probabilities. (These probabilities can be found in the legend of the corresponding figures.)

It is interesting to note that the combined density of the ST schemes is essentially~$0$ for up to about $x\approx 0.684$ and then increases sharply. The ST density function used by \SV\ and by Karger et al.\cite{KKSTY04} exhibit somewhat similar behavior. (In~\cite{SV14} the density is small, but not~$0$, for $x\le 0.61$. In~\cite{KKSTY04}, the density is $0$ up to $6/11\approx 0.54545$ and is then constant.) We have no explanation at present for the `bumpy' behavior of the ST density for $x\in[0.8,1]$. We think that this behavior is not essential and can be compensated by suitably changing the density functions of the other rounding schemes.

It is interesting to note that the the first KT function is used with a probability of about $0.12$, and the second with a probability of about $0.09$. Also, each of the 8 KT density functions shown has two local maxima. We hope that further experiments will shed more light on the behavior of KT schemes.

The IT and DT density functions shown seem to be better behaved. They mostly have `triangular' shape and in particular are nonzero only on a relatively small sub-interval of $[0,1]$. It would be interesting to see whether they could be replaced by step functions. (We currently only allow piecewise linear density functions, with some minimum nonzero separation between two breakpoints.)

We refer to the collection of all basic rounding schemes from each one of the families ST, KT, IT and DT as a \emph{component} of the whole algorithm. Heatmaps showing the cut densities of the four components, and of the whole algorithms, on simplex points of the form, $(u_1,u_2,c,c,\ldots)$, $(u_1,0,u_3,c,c,\ldots)$ and $(u_1,u_2,u_3,0,\ldots,0)$ are shown in Figure~\ref{F-heatmaps}. In points of the first form we have $c=(1-u_1-u_2)/(k-2)$, where $k\to\infty$, and similarly for points of the second form. Simplex points of these forms are among the hardest points for the algorithm, i.e., points in which the combined cut density function attains some of its highest values. Most of the other hard points are of the form $(u_1,u_2,u_3,c,c,\ldots)$, where $c=(1-u_1-u_2-u_3)/(k-3)$, but it is harder to plot them.

A closer look at the densities obtained by the whole algorithm is given in Figure~\ref{F-closeup}. It is quite remarkable how the combined algorithm manages to balance the four components, each with a vastly different behavior, and obtain a density that is almost constant on most of the problematic simplex points.

\section{Verification Algorithm}\label{sec:verification}

We now present our verification algorithms for general $k$ and for small fixed $k$. We also provide some information on implementation details. The verification algorithms are very general and work as long as the mixture $\mathcal{R}$ is  \emph{$\alpha$-compliant} for some $\alpha>0$. (See Definition~\ref{D:compliance}.)

\subsection{Verification for General $k$}
We first focus on the case of general $k$. Let $\alpha \in (0, 1]$, and $\mathcal{R}$ be an $\alpha$-compliant mixture.
For any $\bu = (u_1, \ldots, u_k) \in \Delta_k$, the cut density achieved by $\mathcal{R}$ on $\bu$ is given by
\[
d^\mathcal{R}_k(\bu) \EQ \E_{R \sim \mathcal{R}} d^R_k(\bu) \;.
\]
Our goal is to provide a rigorous verification that for some $d_0 > 0$,
\[
d^\mathcal{R}_k(\bu) \LE d_0\quad, \quad\forall k \geq 2\;,\; \bu \in \Delta_k \;.
\]

As mentioned earlier, our verification algorithm works with tuples of intervals which represent sets of prefixes of simplex points. We formalize this in the following definition.
\begin{definition}
    Let $\bU = (U_1, \ldots, U_\ell)$, where $\ell \geq 2$ and each $U_i$ is an interval within $[0, 1]$. For any simplex point $\bu \in \Delta_k$ where $k \geq \ell$, we say that $\bu$ is \emph{compatible} with $\bU$ if $u_i \in U_i$ for every $1 \leq i \leq \ell$. Furthermore, for any $\alpha > 0$, we say that $\bu$ is \emph{$(\ell, \alpha)$-light}, if $u_i \leq \alpha$ for every $i \in \{\ell + 1, \ldots, k\}$.
\end{definition}

For $\bu = (u_1,u_2,\ldots,u_\ell)$, consider the function $\faked^\mathcal{R}(\bu) = \E_{R \sim \mathcal{R}} \faked^R(\bu)$ where we define
\begin{equation*}
    \faked^R(\bu) \EQ \lim_{k \to \infty}d^{R}_{k}(u_1,u_2,\ldots,u_\ell,\underbrace{c,c,\ldots,c}_{k-\ell})\quad, \quad c=\frac{1-\sum_{i=1}^\ell u_i}{k-\ell}\;,
\end{equation*}
 if $R$ is an IT or GKT scheme, and
\begin{equation*}
    \faked^R(\bu) \EQ  f(u_1) \cdot \prod_{i:u_i\ge u_1}(1-F(u_1,u_i)) + f(u_2) \cdot \prod_{i:u_i\ge u_2}(1-F(u_2,u_i))
\end{equation*}
if $R$ is an DT scheme, and finally
\begin{equation*}
    \faked^R(\bu) \EQ   \frac{f(u_{\min})}{2 + \left|\{i \in [3, \ell] \mid u_i > u_{\min}\}\right|} + \frac{f(u_{\max})}{1 + \left|\{i \in [3, \ell] \mid u_i > u_{\max}\}\right|}
\end{equation*}
where $u_{\min} = \min\{u_1, u_2\}$ and $u_{\max} = \max\{u_1, u_2\}$ if $R$ is an ST scheme. We show that $\faked^R(\bu)$ is an upper bound for the cut density of $R$ at any point $\bu'$ which has $\bu$ as prefix and whose remaining coordinates are all at most $\alpha$:
\begin{lemma}\label{lem:upper_bound_general_k}
    Let $R$ be an IT, GKT, DT, or ST scheme in the support of an $\alpha$-compliant mixture~$\mathcal{R}$. Then, for every $\bu = (u_1, \ldots, u_\ell)$ where $\ell \geq 2$, we have 
    \[
    \faked^R(\bu) \GE \sup_{k \geq \ell}\,\sup_{\substack{0 \leq u_{\ell + 1}, \ldots, u_k \leq \alpha \\ \sum_{i = 1}^k u_i = 1}} d_k^R(u_1, \ldots, u_k)\;.
    \]
\end{lemma}
\begin{proof}
    We need to show that for every $k \geq \ell$ and every $0 \leq u_{\ell + 1}, \ldots, u_k \leq \alpha$ such that $\sum_{i = 1}^k u_i = 1$, we have
    \begin{equation}\label{eq:prefix_upper_bound}
    \faked^R(\bu) \GE d_k^R(u_1, \ldots, u_k)\;.
    \end{equation}
    If $R$ is an IT or GKT scheme, then \eqref{eq:prefix_upper_bound} follows from Lemma~\ref{L-IT-prefix} and Lemma~\ref{L-KT-prefix} respectively. If $R$ is a DT scheme, then by Corollary~\ref{DT-density}, we have for every $0 \leq u_{\ell + 1}, \ldots, u_k \leq \alpha$
    \begin{align*}
        &\, d^{R}_k(u_1,u_2,\ldots,u_k) \\
        \EQ &\, \sum_{j = 1}^2 f(u_j) \cdot \left(\prod_{\substack{i:u_i\ge u_j \\ 1 \leq i \leq k}}(1-F(u_j,u_i))-\prod_{\substack{i \ne j \\ 1 \leq i \leq k}}F(\max\{u_j,u_i\},1) \right)  \\
        \LE &\,  f(u_1) \cdot \prod_{\substack{i:u_i\ge u_1 \\ 1 \leq i \leq k}}(1-F(u_1,u_i)) + f(u_2) \cdot \prod_{\substack{i:u_i\ge u_2 \\ 1 \leq i \leq k}}(1-F(u_2,u_i)) \\
        \LE &\, f(u_1) \cdot \prod_{\substack{i:u_i\ge u_1 \\ 1 \leq i \leq \ell}}(1-F(u_1,u_i)) + f(u_2) \cdot \prod_{\substack{i:u_i\ge u_2 \\ 1 \leq i \leq \ell}}(1-F(u_2,u_i)) = \faked^{R}(\bu)\;.
    \end{align*}
    Finally, if $R$ is an ST scheme, assuming $u_1 \leq u_2$ without loss of generality, we have by Lemma~\ref{ST-density-1} and Lemma~\ref{ST-density-2}
    \begin{align*}
        &\, d^{R}_k(u_1,u_2,\ldots,u_k) \\
        \EQ &\, \frac{f(u_1)}{|\{i \in [k] \mid u_i\geq u_1\}|} + \frac{f(u_2)}{1 + |\{i \in [k] \mid u_i >  u_2\}|} \\
        \LE &\, \frac{f(u_1)}{2 + |\{i \mid u_i > u_1, 3 \leq i \leq k\}|} + \frac{f(u_2)}{1 + |\{i \mid u_i >  u_2, 3 \leq i \leq k\}|} \\
        \LE &\, \frac{f(u_1)}{2 + |\{i \mid u_i > u_1, 3 \leq i \leq \ell\}|} + \frac{f(u_2)}{1 + |\{i \mid u_i >  u_2, 3 \leq i \leq \ell\}|} = \faked^R(\bu)\;.
    \end{align*}
    This completes the proof.
\end{proof}

\begin{definition}
    $\textsc{BoundDensity}^{\mathcal{R}}(\bU)$ is an interval arithmetic implementation for $\faked^{\mathcal{R}}(\bu)$ which satisfies the following guarantee: if $\bU = (U_1, \ldots, U_\ell)$ and $D = \textsc{BoundDensity}^{\mathcal{R}}(\bU)$ \footnote{Note that $\textsc{BoundDensity}^{\mathcal{R}}(\bU)$ returns an interval.}, then
    for every $\bu = (u_1, \ldots, u_\ell)$ compatible with $\bU$, $\faked^\mathcal{R}(\bu) \leq \sup(D)$. 
\end{definition}

\begin{remark}
    In the usual interval arithmetic guarantee, we would have $\faked^\mathcal{R}(\bu) \in D$. Here we only obtain an upper bound (which is sufficient for our purpose) because of the evaluation of $\faked^{\ST(f)}(\bu)$, whose expression involves comparison between coordinates. In our interval arithmetic implementation, the comparison between two intervals $U_i$ and $U_j$, say $U_i < U_j$, is only true when $u_i < u_j$ for every $u_i \in U_i$ and $u_j \in U_j$, so we will often undercount the denominators in the expression of $\faked^{\ST(f)}(\bu)$.   
\end{remark}

We will also make the following assumption on the precision of interval arithmetic. This assumption is only needed to prove termination of our verification algorithms.
\begin{assumption}\label{assumption:ia_precision}
    For every $\eps > 0$ and integer $\ell \geq 2$, there exists some $\delta = \delta(\eps, \ell) > 0$ such that the following holds: for any $\bU = (U_1, \ldots, U_\ell)$, where $\ell \geq 2$ and each $U_i$ is an interval within $[0, 1]$, if $|U_i| \leq \delta$ for every $i \in \{1, 2, \ldots \ell\}$, then for every $d \in D = \textsc{BoundDensity}^{\mathcal{R}}(\bU)$, there exists some $\bu = (u_1, \ldots, u_\ell)$ that is compatible with $\bU$ such that $|\faked^\mathcal{R}(\bu) - d|\leq  \eps$.
\end{assumption}

Informally, the assumption says that by making the input intervals sufficiently small, we may get arbitrarily good precision for the output interval. This, of course, would require arbitrarily good precision. (An interval arithmetic implementation is usually not restricted by the machine precision as it can have its own implementation of arithmetic operations.~\cite{fousse2007mpfr}) In actual runs we only need this for $\eps > 10^{-4}$. 

\begin{algorithm}
    \caption{The verification algorithm for general $k$}\label{alg:general_k}
    \begin{algorithmic}[1]
        \State \textbf{Input:} $\bU = (U_1, \ldots, U_\ell)$, where $\ell \geq 2$ and each $U_i$ is an interval within $[0, 1]$; target density $d_0$
        \Procedure{Verify}{$\bU, d_0$}
            \State $D \gets \textsc{BoundDensity}^{\mathcal{R}}(\bU)$ \label{line:compute_D}
            \If {$\min(D) > d_0$} 
                \State \Return $\false$\label{Line:min_check}
            \EndIf
            \If {$\max(D) \leq d_0$}
                \If {$1 - \min(\sum_{i = 1}^\ell U_i) < \alpha$}\label{line:check_b}
                    \State \Return $\true$
                \Else 
                    \State $U_{\ell + 1} \gets [\alpha, 1 - \min(\sum_{i = 1}^\ell U_i)]$ \label{line:add_U}
                    \State $\bU' = (U_1, \ldots, U_{\ell + 1})$
                    \State \Return \textsc{Verify}$(\bU', d_0)$ \label{line:bU'}
                \EndIf 
            \EndIf
            \State Let $U_i$ be a longest interval in $\bU$
            \State Split $U_i$ into two equal-length subintervals $U_i^\leftsf$ and $U_i^\rightsf$
            \State $\bU^{\leftsf} \gets (U_1, \ldots, U_{i - 1}, U_i^{\leftsf}, U_{i + 1}, \ldots, U_\ell)$
            \State $\bU^{\rightsf} \gets (U_1, \ldots, U_{i - 1}, U_i^{\rightsf}, U_{i + 1}, \ldots, U_\ell)$
            \State \Return \textsc{Verify}$(\bU^\leftsf, d_0)$ $\wedge$ \textsc{Verify}$(\bU^\rightsf, d_0)$
        \EndProcedure
    \end{algorithmic}
\end{algorithm}

We now present and analyze our verification algorithm for general $k$. The pseudocode for the algorithm can be found in Algorithm~\ref{alg:general_k}. In the following analysis, when we use the phrase \emph{the execution of }$\textsc{Verify}(\bU, d_0)$, we refer to the entire computation including all recursive calls made by it and its subroutines.
\begin{lemma}\label{lem:true_covered}
    Let $\bU = (U_1, \ldots, U_\ell)$, where $\ell \geq 2$ and each $U_i$ is an interval within $[0, 1]$. Assume that $\textsc{Verify}(\bU, d_0)$ returns $\true$ for some $d_0 > 0$. Let $\bu \in \Delta_k$ be compatible with $\bU$, and let $p$ be the minimum value in $\{\ell, \ldots, k\}$ such that $\bu$ is $(p, \alpha)$-light. Then, there must be some $\bU^{(p)} = (U_1^{(p)}, \ldots, U_p^{(p)})$ with which $\bu$ is compatible such that the execution of $\textsc{Verify}(\bU, d_0)$ contains a recursive call to $\textsc{Verify}(\bU^{(p)}, d_0)$.
\end{lemma}
\begin{proof}
    We prove this by induction on $p$. If $p = \ell$ then we can take $\bU^{(p)} = \bU$. Assume $p > \ell$. Let $\bu' = (u_1', \ldots, u_{k'}')$ be such that $u_i' = u_i$ for every $1 \leq i \leq p - 1$ and $u_p' = \cdots = u_{k'}' = \frac{1 - \sum_{i = 1}^{p-1} u_i}{k' - p + 1}$. We may choose $k'$ to be sufficiently large so that $\bu'$ is $(p - 1, \alpha)$-light. By the inductive hypothesis, there exists a recursive call $\textsc{Verify}(\bU^{(p-1)}, d_0)$ such that $\bU^{(p-1)}$ has length $p-1$, $\bu'$ is compatible with~$\bU^{(p-1)}$, and the execution of $\textsc{Verify}(\bU, d_0)$ contains a recursive call to $\textsc{Verify}(\bU^{(p-1)}, d_0)$. We take $\textsc{Verify}(\bU^{(p-1)}, d_0)$ to be the last such recursive call, which exists since $\textsc{Verify}(\bU, d_0)$ halts (it returns $\true$). For this call, the if-else block starting on Line~\ref{line:check_b} must be executed, for otherwise the program would create two recursive calls $\textsc{Verify}(\bU^{(p-1), l}, d_0)$ and $\textsc{Verify}(\bU^{(p-1), r}, d_0)$, one of which $\bu'$ must be compatible with. Now, the if condition on Line~\ref{line:check_b} does not hold, since 
    \[
        1 - \min\left(\sum_{i = 1}^{p-1}U^{(p-1)}_i\right) \GE 1 - \sum_{i = 1}^{p-1}u_i \GE u_{p} \GE \alpha\;.
    \]
    Therefore, Line~\ref{line:bU'} will be called with some $\bU'$ of length $p$ which $\bu$ is compatible with, so we may take $\bU^{(p)} = \bU'$. 
\end{proof}

\begin{theorem}\label{thm:ia_correctness}
    Let $\bU = (U_1, \ldots, U_\ell)$, where $\ell \geq 2$ and each $U_i$ is an interval within $[0, 1]$ and $d_0 > 0$. Then we have the following correctness guarantee:
    \begin{itemize}
        \item[(1)] Assume that $\textsc{BoundDensity}^{\mathcal{R}}(\bU)$ satisfies Assumption~\ref{assumption:ia_precision}. If there exists some $\eps > 0$ such that $\faked^\mathcal{R}(\bu) \leq d_0 - \eps$ for every $\bu$ compatible with $\bU$, then $\textsc{Verify}(\bU, d_0)$ returns $\true$.
        \item[(2)] If $\textsc{Verify}(\bU, d_0)$ returns $\true$, then $\faked^\mathcal{R}(\bu) \leq d_0$ for every $\bu$ compatible with $\bU$. In particular, $d^\mathcal{R}_k(\bu) \leq d_0$ for every $k \geq \ell$ and $\bu \in \Delta_k$ compatible with $\bU$.
    \end{itemize}
\end{theorem}
\begin{proof}
    For part (1), consider the execution of $\textsc{Verify}(\bU, d_0)$ before any recursive call is made. Let~$D$ be the output of $\textsc{BoundDensity}^{\mathcal{R}}(\bU)$ which we obtain on Line~\ref{line:compute_D}. Then by interval arithmetic,
    \[
    \{\faked^\mathcal{R}(\bu) \mid \bu = (u_1, \ldots, u_\ell) \text{ compatible with } \bU\} \subseteq D \;.
    \]
    Since $\faked^\mathcal{R}(\bu) \leq d_0 - \eps$ for every $\bu$ compatible with $\bU$, we have $\min(D) \leq \faked^\mathcal{R}(\bu) < d_0$, so the call does not return $\false$ on Line~\ref{Line:min_check}. The same logic applies to all subsequent recursive calls, so the program never returns $\false$. Now it is sufficient to argue that the program will only make finitely many recursive calls, so it must halt and return $\true$. 
    
    Since each new coordinate we add at Line~\ref{line:add_U} has minimum value $\alpha > 0$, we can add at most $1/\alpha$ coordinates. Assume for the sake of contradiction that $\textsc{Verify}(\bU, d_0)$ never halts, then for some $p \leq \ell + 1 / \alpha$ the execution of $\textsc{Verify}(\bU, d_0)$ must contain infinitely many calls of the form $\textsc{Verify}(\bU', d_0)$ where $\bU'$ contains $p$ intervals. Choose a $p$ to be maximal with this property (so that only finitely many calls occur with $\bU'$ containing more than $p$ intervals). Then there must be some $\bU'$ of length $p$ satisfying the following:
    \begin{itemize}
        \item the execution of $\textsc{Verify}(\bU', d_0)$ does not halt.
        \item $\textsc{Verify}(\bU', d_0)$ is called during the execution of $\textsc{Verify}(\bU, d_0)$.
        \item For any $\bU''$ whose length is greater than $p$, the execution of $\textsc{Verify}(\bU', d_0)$ never calls $\textsc{Verify}(\bU'', d_0)$ as a subroutine.
    \end{itemize}

    Let $\delta = (\eps / 2, p) > 0$ be guaranteed as in Assumption~\ref{assumption:ia_precision}. In the recursive tree for the execution of $\textsc{Verify}(\bU', d_0)$, the maximal length of any interval is decreased by at least half if the depth is increased by $p$, we may additionally assume that 
    \begin{itemize}
        \item $\bU' = (U_1', \ldots, U_p')$ where $|U_i'| < \delta$ for any $i \in \{1, \ldots, p\}$. 
    \end{itemize}
    Now we let $d = \max(D')$ where $D' = \textsc{BoundDensity}^{\mathcal{R}}(\bU')$, then by Assumption~\ref{assumption:ia_precision} we can find some $\bu$ compatible with $\bU'$ such that $\faked^\mathcal{R}(\bu) \geq \max(D) - \eps / 2$. It follows that $\max(D) \leq \faked^\mathcal{R}(\bu) + \eps / 2 \leq d_0 - \eps / 2$, so the if-else statement starting on Line~\ref{line:check_b} will be executed when we call $\textsc{Verify}(\bU', d_0)$. But this means either $\textsc{Verify}(\bU', d_0)$ returns $\true$ or it makes another recursive call with some $\bU''$ with length greater than $p$, contradicting our assumption on $\bU'$.
    
    For part (2), let $\bu = (u_1, \ldots, u_k)$ be compatible with $\bU$. By symmetry, we may permute some of the indices greater than $\ell$ and assume that $u_{\ell + 1} \geq \cdots \geq u_k$. Find the minimum $p \in \{\ell, \ell + 1, \cdots, k\}$ such that $\bu$ is $(p, \alpha)$-light. By Lemma~\ref{lem:true_covered}, there exists a recursive call $\textsc{Verify}(\bU^{(p)}, d_0)$. We take the last such call. Then, the if-else block starting on Line~\ref{line:check_b} must be executed, which means in the previous line $\max(D) \leq d_0$ holds. So we have $\faked^\mathcal{R}(\bu) \leq \max(D) \leq d_0$.
\end{proof}

\begin{remark}
We remark that the ``$\eps$-room'' is needed in the first part because if $d(\bu) \leq d_0$ but the equality is achieved at some $\bu_0$, then due to the inherent loss in interval arithmetic, we can never distinguish with certainty whether $d(\bu_0) = d_0$ or $d(\bu_0)$ is slightly above $d_0$ using interval arithmetic calculations.
\end{remark}

\subsection{Verification for Fixed Finite $k$}

The verification for fixed $k$ (see Algorithm~\ref{alg:finite_k} for pseudocode) is similar to that of general $k$, with a few differences which we now explain in detail. Again, we assume that we are given $\alpha > 0$ and an $\alpha$-compliant mixture $\mathcal{R}$.

We first give the new evaluation function $\textsc{BoundDensity}_k(\bU)$. Given as input $\bU = (U_1, \ldots, U_\ell)$, where $2 \leq \ell \leq k$, $\textsc{BoundDensity}_k(\bU)$ computes, for every $\bu = (u_1, \ldots, u_\ell)$ compatible with $\bU$, the function $\faked^\mathcal{R}_k(\bu) = \E_{R \sim \mathcal{R}} \faked^R_k(\bu)$ where 
\begin{equation*}
    \faked^R_k(\bu) = d^{\mathcal{R}}_{k}(u_1,u_2,\ldots,u_\ell,\underbrace{c,c,\ldots,c}_{k-\ell}) \quad, \quad c=\frac{1-\sum_{i=1}^\ell u_i}{k-\ell}\;,
\end{equation*}
if $R$ is an IT or GKT scheme, and
\begin{equation*}
    \faked^R_k(\bu) \EQ \sum_{j = 1}^2 f(u_j) \cdot \left(\prod_{i:u_i\ge u_j}(1-F(u_j,u_i))-\left(\prod_{i\ne j}F(\max\{u_j,u_i\},1)\right) F(\max\{u_j, \alpha\}, 1)^{k - \ell}\right)
\end{equation*}
if $R$ is an DT scheme, and
\begin{equation*}
    \faked^R_k(\bu) \EQ   \frac{f(u_{\min})}{2 + \left|\{i \in [3, \ell] \mid u_i > u_{\min}\}\right|} + f(u_{\max}) \cdot \min\left\{1 - \frac{1}{k}, \frac{1}{1 + \left|\{i \in [3, \ell] \mid u_i > u_{\max}\}\right|}\right\}
\end{equation*}
where $u_{\min} = \min\{u_1, u_2\}$ and $u_{\max} = \max\{u_1, u_2\}$ if $R$ is an ST scheme.

\begin{lemma}\label{lem:upper_bound_small_k}
    Let $R$ be an IT, GKT, DT, or ST scheme in the support of an $\alpha$-compliant mixture~$\mathcal{R}$. For every $\bu = (u_1, \ldots, u_\ell)$ where $2 \leq \ell \leq k$, we have 
    \[
    \faked_k^R(\bu) \GE \max_{\substack{0 \leq u_{\ell + 1}, \ldots, u_k \leq \alpha \\ \sum_{i = 1}^k u_i = 1}} d_k^R(u_1, \ldots, u_k) \;.
    \]
\end{lemma}
\begin{proof}
    For IT and GKT, the lemma follows from Lemma~\ref{L-IT-prefix} and Lemma~\ref{L-KT-prefix} respectively. If $R$ is a DT scheme, then by Corollary~\ref{DT-density}, we have for every $0 \leq u_{\ell + 1}, \ldots, u_k \leq \alpha$
    \begin{align*}
        &\, d^{R}_k(u_1,u_2,\ldots,u_k) \\
        \EQ &\, \sum_{j = 1}^2 f(u_j) \cdot \left(\prod_{\substack{i:u_i\ge u_j \\ 1 \leq i \leq k}}(1-F(u_j,u_i))-\prod_{\substack{i \ne j \\ 1 \leq i \leq k}}F(\max\{u_j,u_i\},1) \right)  \\
        \LE &\, \sum_{j = 1}^2 f(u_j) \cdot \left(\prod_{\substack{i:u_i\ge u_j \\ 1 \leq i \leq l}}(1-F(u_j,u_i))-\prod_{\substack{i \ne j \\ 1 \leq i \leq \ell}}F(\max\{u_j,u_i\},1) \prod_{\substack{i \ne j \\ \ell+1 \leq i \leq k}}F(\max\{u_j,u_i\},1) \right)  \\
        \LE &\, \sum_{j = 1}^2 f(u_j) \cdot \left(\prod_{\substack{i:u_i\ge u_j \\ 1 \leq i \leq l}}(1-F(u_j,u_i))-\prod_{\substack{i \ne j \\ 1 \leq i \leq \ell}}F(\max\{u_j,u_i\},1) \cdot F(\max\{u_j, \alpha\}, 1)^{k - \ell}\right) \EQ \faked_k^{R}(\bu)\;.
    \end{align*}
    The last inequality follows since $0 \leq u_i \leq \alpha$ for every $\ell + 1\leq i \leq k$ so $F(\max\{u_j, \alpha\}, 1) = 1 - F(\max\{u_j, \alpha\}) \leq 1 - F(\max\{u_j, u_i\}) = F(\max\{u_j, u_i\}, 1)$. Finally, if $R$ is an ST scheme, assuming $u_1 \leq u_2$ without loss of generality, we have by Lemma~\ref{ST-density-1} and Lemma~\ref{ST-density-2}
    \begin{align*}
        &\, d^{R}_k(u_1,u_2,\ldots,u_k) \\
        \EQ &\, \frac{f(u_1)}{|\{i \in [k] \mid u_i\geq u_1\}|} + f(u_2)\cdot \min\left\{1 - \frac{1}{k}, \frac{1}{1 + |\{i \in [k] \mid u_i >  u_2\}|}\right\} \\
        \LE &\, \frac{f(u_1)}{2 + |\{i \mid u_i > u_1, 3 \leq i \leq k\}|} + f(u_2)\cdot \min\left\{1 - \frac{1}{k}, \frac{1}{1 + |\{i \mid u_i >  u_2, 3 \leq i \leq k\}|}\right\} \\
        \LE &\, \frac{f(u_1)}{2 + |\{i \mid u_i > u_1, 3 \leq i \leq \ell\}|} + f(u_2)\cdot \min\left\{1 - \frac{1}{k}, \frac{1}{1 + |\{i \mid u_i >  u_2, 3 \leq i \leq \ell\}|}\right\} \EQ \faked^R_k(\bu)\;.
    \end{align*}
    This completes the proof.
\end{proof}

By Lemma~\ref{lem:upper_bound_small_k}, we know that if $d_0 \geq \faked^R_k(\bu)$, then the density will not exceed $d_0$ unless we add a coordinate greater than $\alpha$, so the recursion structure from the general $k$ case still applies here. On the other hand, when $d_0 < \faked^R_k(\bu)$, it is possible that we overestimated the density (note that $\faked^R_k(\bu)$ is only an upper bound on the density) since different rounding schemes may achieve their worst-case density on different points, so we need to check if adding one more coordinate decreases the density.  This explains Line~\ref{Line:min_check_finite_k}, where we do not return $\false$ unless $\ell = k$, and Line~\ref{Line:add_new_interval_finite_k}, where we need to check if adding a new coordinate helps ($\eps > 0$ here does not affect correctness, though it does affect the actual running time) . The final difference is on Line~\ref{line:check_alpha_finite_k}, where compared to general~$k$ we also have the halting condition $\ell = k$, which is self-evident. The formal correctness proof is similar to the general $k$ case, and we omit it here.

\begin{algorithm}
    \caption{The verification algorithm for finite $k$}\label{alg:finite_k}
    \begin{algorithmic}[1]
        \State \textbf{Input:} $\bU = (U_1, \ldots, U_\ell)$, where $2 \leq \ell \leq k$ and each $U_i$ is an interval within $[0, 1]$; target density $d_0$; some parameter $\eps > 0$.
        \Procedure{Verify}{$\bU, d_0$}
            \State $D \gets \textsc{BoundDensity}_k(\bU)$ 
            \If {$\min(D) > d_0$ and $\ell = k$} \label{Line:min_check_finite_k}
                \State \Return $\false$
            \EndIf
            \If {$\max(D) \leq d_0$}
                \If {$1 - \min(\sum_{i = 1}^\ell U_i) < \alpha$ or $\ell = k$}\label{line:check_alpha_finite_k}
                    \State \Return $\true$
                \Else 
                    \State $U_{\ell + 1} \gets [\alpha, 1 - \min(\sum_{i = 1}^\ell U_i)]$ 
                    \State $\bU' = (U_1, \ldots, U_{\ell + 1})$
                    \State \Return \textsc{Verify}$(\bU', d_0)$
                \EndIf 
            \EndIf
            \State Let $U_i$ be a longest interval in $\bU$ \Comment{If we are here, then $\max(D) > d_0$}
            \If {$|U_i| \leq \eps$ and $\ell < k$} \Comment{Add new interval if max interval length is small} \label{Line:add_new_interval_finite_k}
            
                \State $U_{\ell + 1} \gets [0, 1 - \min(\sum_{i = 1}^\ell U_i)]$ 
                \State $\bU' = (U_1, \ldots, U_{\ell + 1})$
                \State \Return \textsc{Verify}$(\bU', d_0)$ 
            \EndIf 
            \State Split $U_i$ into two equal-length subintervals $U_i^\leftsf$ and $U_i^\rightsf$
            \State $\bU^{\leftsf} \gets (U_1, \ldots, U_{i - 1}, U_i^{\leftsf}, U_{i + 1}, \ldots, U_\ell)$
            \State $\bU^{\rightsf} \gets (U_1, \ldots, U_{i - 1}, U_i^{\rightsf}, U_{i + 1}, \ldots, U_\ell)$
            \State \Return \textsc{Verify}$(\bU^\leftsf, d_0)$ $\wedge$ \textsc{Verify}$(\bU^\rightsf, d_0)$
        \EndProcedure
    \end{algorithmic}
\end{algorithm}

\subsection{Implementation and Verification Details}\label{subsec:ver-impl}

Now that we have discussed the interval arithmetic verification algorithms, we now give a few details about their implementation.
The core of the interval arithmetic verification is implemented using the Arb library~\cite{johansson2017arb} as part of the FLINT C libraries~\cite{flint}. Arb itself builds on the GNU MP~\cite{Granlund24} and MPFR~\cite{fousse2007mpfr} libraries for its arbitrary precision floating point arithmetic. To implement the bound in Lemma~\ref{L-IT-asymptotic}, we make sure of Arb's implementation of hypergeometric functions~\cite{johansson2019computing}. We also make use of a recently made C++ wrapper for a fragment of Arb~\cite{BHZ24}.

\begin{table}
\begin{center}
\begin{tabular}{c@{\hspace{2em}}c@{\hspace{1em}}c@{\hspace{1em}}c@{\hspace{1em}}c@{\hspace{1em}}c}
$k$ & Verified Ratio & N & Cores & Wall Time (hr) & Total Work (hr)\\[3pt]
\hline
\noalign{\vskip 3pt}\noalign{\vskip 3pt}
$4$ & \ourAPXd & $20000$ & $192$ & 10.5 & 1646\\
$5$ & \ourAPXe & $8000$ & $128$ & 21.9 & 2762\\
$6$ & \ourAPXf & $8000$ & $128$ & 11.7 & 1461\\
$7$ & \ourAPXg & $2000$ & $128$ & 2.4 & 296\\
$8$ & \ourAPXh & $2000$ & $128$ & 3.0 & 370\\
$9$ & \ourAPXi & $2000$ & $128$ & 4.0 & 456\\
$10$ & \ourAPXj & $2000$ & $128$ & 5.3 & 487\\
\hline
\noalign{\vskip 3pt}\noalign{\vskip 3pt}
any & \ourAPX & $4000$ & $192$ & 5.3 & 967\\
\end{tabular}
\end{center}
\caption{This table summarizes the results of our verifications of each of our results. Here ``$N$'' is the number of tasks the verification was split into, ``Wall Time'' refers to amount of physical time that passed to run the verification, and ``Total Work'' refers to the total amount of time used by all the cores.}\label{table:verification}
\end{table}

\paragraph{Parallel Verification.} In order make the verification computationally feasible, we take advantage of an intuitive source of parallelism. Recall that it suffices to $\textsc{Verify}(\bU, d_0)$, where $d_0$ is our target approximation ratio and $\bU = ([0,1])$. To parallelize this task, we pick a natural number $N \in \mathbb N$ (e.g., $N=2000$) and consider $N$ tasks $\bU_i := ([\frac{i-1}{N},\frac{i}{N}])$ for $i \in [N] := \{1,2\hdots, N\}$.

For each of our verified bounds, we use GNU Parallel~\cite{Tange2011} to deploy these $N$ tasks across machines at the University of California, Berkeley's Savio3 HPC cluster using Intel Xeon Skylake chips with 32 cores.\footnote{Some machines used had 40 cores, although in such cases GNU Parallel was configured to only take advantage of 32 of them.} We summarize the results of the verifications in Table~\ref{table:verification}. In total, over 8000 core-hours we used to verify the results in this paper.

\section{Discussion}\label{sec:discussion}

In this section, we give further discussion on how our new rounding schemes are related to the goal of eventually finding truly optimal rounding schemes for \MC. It essentially follows from arguments given in Karger et al.~\cite{KKSTY04}, combined with result of Manokaran et al.~\cite{MNRS08}, that the best approximation ratio that can be achieved for the \MC\ problem with~$k$ terminals, under the Unique Games Conjecture (UGC), is equal to the value of the following infinite $0$-sum game. An \emph{edge player} chooses two distinct points $\bu\ne\bv\in\Delta_k$, while a \emph{cut player} simultaneously chooses a $k$-cut, or more precisely, a measurable labeling $c:\Delta_k\to[k]$ satisfying $c(\be_i)=i$, for $i\in[k]$. The outcome of the game is $\frac{{[c(\bu)\ne c(\bv)]}}{\frac12 \|\bu-\bv\|_1}$. (Here $[c(\bu)\ne c(\bv)]$ is equal to~$1$ if $c(\bu)\ne c(\bv)$ and to~$0$ otherwise.) The cut player tries to minimize this value while the edge player tries to maximize it. A mixed strategy of the edge player is a distribution over pairs $(\bu,\bv)\in \Delta_k\times \Delta_k$.  A mixed strategy of the cut player is a distribution over labelings $c:\Delta_k\to[k]$.

The set of mixed strategies of the edge player, though infinite, is relatively manageable. In particular, it is compact. It is also not difficult to show that the cut player has an optimal mixed strategy that is symmetric, in which case the edge player may restrict herself to infinitesimal edges that are $(1,2)$-aligned, so her mixed strategy reduces to a distribution over $\Delta_k$.

By contrast, the space of mixed strategies available to the cut player is far less manageable. In particular, suitable measurability assumptions are needed to ensure that the game is well defined and that it has a value.

Our computation approach approximates the value of this infinite game using mixed strategies with finite support. However, the strategies in the cut player's support are not pure, i.e., deterministic, labelings, but rather basic rounding schemes, i.e., continuous distributions of fairly restricted forms over labelings.\footnote{It is not difficult to check that any distribution that gives a single labeling $c_0:\Delta_k\to[k]$ a positive probability has an infinite value.} We currently use four different families of basic rounding schemes: the newly introduced Generalized Kleinberg-Tardos (GKT), Single Threshold (ST), Independent Thresholds (IT) and Descending Thresholds (DT).

We say that a labeling $c:\Delta_k\to[k]$ is \emph{convex} if $c^{-1}(i)$ is convex, for every $i\in[k]$.
We say that a labeling $c:\Delta_k\to[k]$ is \emph{polytopal} if $c^{-1}(i)$ is a polytope, for every $i\in[k]$. Furthermore, a labeling $c:\Delta_k\to[k]$ is \emph{facet-parallel polytopal} if each facet of each one of the polytopes $c^{-1}(i)$, for $i\in[k]$, is parallel to one of the facets of the simplex. Each of ST, IT and DT uses only facet-parallel polytopal labeling. In fact, they only use what Karger et al.~\cite{KKSTY04} call \emph{side-parallel cuts (sparcs)}, defined by a permutation~$\sigma$ on the terminals and a sequence of thresholds, as in the definition of IT schemes. The cuts used by EC are polytopal, but not of the other forms. The cuts obtained using the simplex transformations of Buchbinder et al. \cite{BSW21}, as far as we can see, are not polytopal.

It is interesting to note that GKT, and even KT, while only branching on conditions of the form $u_i\ge t$, may generate cuts that are not convex. For example, if $k=3$, the sequence of terminals considered is $1,2,1,3,\dots$, and the sequence of thresholds is $0.2,0.25,0.1,0,\dots$, then $c^{-1}(1)=\{(u_1,u_2,u_3)\in\Delta_3 \mid u_1\ge 0.2 \lor (u_1\ge 0.1 \land u_2<0.25)\}$. If $\bu_1=(0.24,0.76,0)$ and $\bu_2=(0.14,0,0.86)$, then $c(\bu_1)=c(\bu_2)=1$ while $c(\frac{1}{2}(\bu_1+\bu_2))=2$.

It would be extremely interesting to know which types of cuts are needed for obtaining optimal approximation algorithms for the \MC\ problem. Are sparcs enough? \KKSTYal\ show that they are enough for $k=3$, but the question is open for $k\ge 4$. Are non-convex cuts needed? If not, why are the GKT family so useful?

Each of the families ST$(f)$, IT$(f)$ and DT$(f)$, while inducing distributions over sparcs, cannot induce an arbitrary distribution over sparcs. To describe such a general distribution we need a density function $f:[0,1]^{k}\to\RR^+$ of a continuous $k$-dimensional random variable. As discussed by \SV, such a distribution should be symmetric, and in particular the marginal distribution induced on any two terminals should be the same. This discussion lead them to define the DT family of rounding schemes.

A smaller family of rounding schemes that at least in the limit can approximate a general distribution over sparcs is what we can call IT$(f_1,f_2,\ldots,f_k)$. This is a version of IT in which each threshold is chosen according to its own distribution. symmetry is assured by the choice of the random permutation. The threshold for the $i$-th terminal in the permutation is chosen using $f_i$. \KKSTYal\ use such rounding schemes to obtain their best ratios for $k=4,5$. More specifically, for some value of~$N$, and for every $i_1,i_2,\ldots,i_k\in [N]$ they use IT$(f_1,f_2,\ldots,f_k)$ where~$f_j$ is the density function of a uniformly random variable on $[\frac{i_j-1}{N},\frac{i_j}{N}]$. They solve a huge linear program to find the best mixture of such schemes. For $k\ge 6$, the values of~$N$ that could be used are too small to provide good approximation ratios.

Mixtures of KT$(f)$, ST$(f)$, IT$(f)$ and DT$(f)$, while unlikely to give optimal algorithms, seem to give substantially improved results. It would be interesting to understand the limits of what can be achieved using them. (In particular, we note that they cannot be used to obtain an optimal algorithm for $k=3$, though they come very close.) It would also be interesting to see whether other interesting families can be added to the mix.

We hope that a closer look of the density functions used for each one of these families will reveal some interesting patterns. Understanding these patterns may lead to \emph{meta rounding schemes} in which we choose a continuous distribution ${\mathcal F}_{\KT}$ over simple density functions to be used by KT schemes, etc. The simple density functions may, for example, be step functions that are nonzero only on two intervals so that a continuous distribution over them would be simple to describe.

\section{Concluding Remarks and Further Directions}\label{sec:concl}

\MC\ is a fundamental and intriguing optimization problem. Although some improved approximation ratios were obtained in this paper, we are still far from obtaining optimal, or close to optimal, approximation algorithms for it. We believe that our computational techniques can be pushed a bit further and we hope to present some further improved results in the final version of this paper. We also plan to take a closer look at the newly discovered rounding schemes and try to gain some insights from them that may potentially lead to more improvements.

The two most interesting setting of the problem at present are the general case with an arbitrary number of terminals, which has been the main focus of previous papers, and the case of $k=4$ terminals. Since optimal approximation algorithms are known for the $k=3$ case~\cite{CT99,KKSTY99,KKSTY04,CCT06}, it is natural to seek a better understanding of the $k=4$ case, though this appears to be a highly non-trivial task.

We conclude the paper with a few directions toward further improving our understanding of \MC\ in both the general and fixed $k$ settings.

\paragraph{Richer Distributions of Rounding Schemes.}

Although our current methods could certainly be refined to give modest improvements to Theorem~\ref{thm:main} and Theorem~\ref{thm:finite}, new ideas are needed to give substantial improvements. In Section~\ref{sec:discussion}, we discuss how even for $k=3$ a mixture of $\ST, \DT, \IT$, and GKT does not appear to be enough to achieve the (known) optimal approximation ratio. As such, more effort needs to be put into designing novel rounding schemes which can help push the approximation ratio of \MC\ even lower. Some candidates for richer families of rounding schemes are described in Section~\ref{sec:discussion}, but care is needed to not make the rounding schemes \emph{too} general so that any optimization runs into the ``curse of dimensionality.''

As such, we ask the following open ended question: can we prove \emph{any} structural characteristics of the truly optimal \MC\ schemes either for $k$ general or $k=4$? For example, does the cut regions of \MC\ need to be convex, or is the non-convexity of (G)KT a necessary feature?

\paragraph{Improved Verification Techniques.}

Due to the complexity of the schemes we studied in this paper, using interval arithmetic verification was essential for giving a rigorous bound on their accuracy. As noted in Table~\ref{table:verification}, the interval arithmetic verifications we conducted requires significant computing resources. We leave improving the runtime of the verification scheme as the subject of future work. The following are a few concrete directions in which this could be possible.

One potential source of improvement is exploring different interval splitting techniques. Right now, given a prefix $\bU = (U_1, \hdots, U_\ell)$, we pick the longest interval $U_i$ to split in our recursion. However, the longest interval may not necessarily be the largest contributor to the error of the rounding scheme. For example, some other interval $U_j$ might lie in a region where the probability density functions have larger derivatives, contributing more error. 

Second, more understanding is needed of how to take advantage of parallelism in the verification. As mentioned in Section~\ref{subsec:ver-impl}, we do a ``uniform'' case split, where each task is based on an interval of the form $(\tfrac{i-1}{N}, \tfrac{i}{N})$. However a more dynamic case split (possibly in more than one dimension) would likely help distribute the workload more accurately.

Finally, another approach is to instead verify \emph{smoother} rounding schemes. One approach toward this would be to understand ``Meta schemes'' (i.e., parameterized mixtures of basic rounding schemes, see Section~\ref{sec:discussion}) and coming up with new analytical formulas for their cut densities. This could greatly decrease the numerical error in the interval arithmetic calculations, thus leading to faster running times.

\paragraph{Lower Bounds.}

While this paper focuses on improved approximation algorithms for \MC, there appears to be room for improvement in terms of integrality gap (and thus UG-hardness) lower bounds for \MC.  As previously mentioned, the current best lower bound for small values of $k$ is $6/(5+\frac{1}{k-1})$ due to Angelidakis, Makarychev, and Manurangsi~\cite{AMM17}, improving on a long-standing bound of $8/(7+\frac{1}{k-1})$ by Freund and Karloff~\cite{FK00}. The main source of the improved analysis in \cite{AMM17} is from looking at a variant of \MC, which they call \emph{Non-opposite} \MC. A precise description is beyond the scope of this article, but the key idea is that by studying \emph{Non-opposite} \MC\ in a fixed-dimensional simplex $\Delta_k$, one can deduce lower bounds on the traditional \MC\ problem for any number of terminals $k' \ge k$. More precisely, Angelidakis, Makarychev, and Manurangsi~\cite{AMM17} adapt the techniques of Cheung, Cunningham, and Tang \cite{CT99,CCT06} and Karger et al.~\cite{KKSTY99,KKSTY04} to tightly understand the behavior of non-opposite cuts in $\Delta_3$, yielding a number of new lower bounds including $1.2$ for general $k$.

However, it appears that even better lower bounds can be found by studying non-opposite cuts in~$\Delta_k$ for $k\ge 4$. In fact, B{\'e}rczi, Chandrasekaran, Kir{\'a}ly, and Madan~\cite{BCKM20}, found an improved integrality gap for $\Delta_4$ with a ratio of $1.20016$, improving the lower bound of \MC\ for general $k$ accordingly. However, their methods do not appear to yield improvements for any (reasonably) small fixed $k$. 

Given this landscape, we believe the next major frontier in the study of \MC\ is to obtain a tight understanding of both traditional cuts and non-opposite cuts in $\Delta_4$. Furthermore, the minimax framework we use in this paper to design algorithms for \MC\ could also be useful in this regard. However, much modification of the techniques is needed, as an adversarial Cut player could design \emph{any} legal cut of $\Delta_4$ rather than just the basic schemes (such as $\KT$, $\ST$, $\IT$ and $\DT$) studied in this paper. We leave further investigation of how to overcome these challenges to future work.

\section*{Acknowledgments}

We acknowledge help from ChatGPT while writing the code for discovering new rounding schemes and while preparing some of the plots. We emphasize that we did \emph{not} use ChatGPT or any other LLM model while writing our verification code. We thank the University of California, Berkeley, for providing the HPC resources used in this project.

\bibliographystyle{alphaurl}
\bibliography{ref}

\newpage
\clearpage
\appendix
\pagestyle{appendix}

\section{The LP Relaxation of \MC}\label{A:LP-relax}

The LP relaxation of \MC\ introduced by \cite{CKR00} is given in Figure~\ref{F-LP-relax}. Let $G=(V,E,w)$, where $w:E\to\RR^+$ be the weighted undirected input graph. We assume that $V=[n]=\{1,2,\ldots,n\}$ and that the $k$ terminals are $[k]=\{1,2,\ldots,k\}$ On the left, it is given as a simplex embedding problem. The relaxation assigns to each vertex $i\in V$ a vector $\bu_i\in \Delta_k$ in the $k$-simplex $\Delta_k=\{ (x_1,x_2,\ldots,x_k) \mid \sum_{i=1}^k x_i=1 , x_1,x_2,\ldots,x_n\ge 0\}$. For a terminal $i\in[k]$ we require that $\bu_i=\be_i$, where $\be_i$ is the $i$-th unit vector. Note that $\be_1,\be_2,\ldots,\be_k$ are the vertices of $\Delta_k$. The objective is to minimize half the weighted sum of the $L_1$-distances between simplex points assigned to adjacent vertices in the graph. This is indeed a relaxation of the problem as any integral solution, i.e., any multiway cut separating the terminals, can be realized by assigning each $\bu_i$ to one of the unit vectors $\be_1,\be_2,\ldots,\be_k$. Note that for every $i\ne j\in [k]$ we have $\frac12\|\be_i-\be_j\|_1=1$.

The simplex embedding program has a clear intuitive meaning, but it is not a linear program, due to the appearance of the nonlinear terms $\|\bu_i-\bu_j\|_1$ in the objective function. However, it is not difficult to see that it is equivalent to the linear program given on the right of Figure~\ref{F-LP-relax}. 

\begin{figure}
\[
\begin{array}{llll}
 \min & \multicolumn{3}{l}{\displaystyle\frac12 \sum_{\{i,j\}\in E} w_{i,j}\|\bu_i-\bu_j\|_1} \\[5pt]
 \text{s.t.} & \bu_i\in \Delta_k &,& i\in V\\
             & \bu_i=\be_i &,& i\in[k] \\
\end{array}
\hspace{1cm}
\begin{array}{llll}
 \min & \multicolumn{3}{l}{\displaystyle\frac12 \sum_{\substack{\{i,j\}\in E\\\ell\in[k]}} w_{ij}d_{i,j,\ell}} \\[5pt]
 \text{s.t.} & d_{i,j,\ell} \ge u_{i,\ell}-u_{j,\ell} &,& \{i,j\}\in E \;,\; \ell\in[k]\\
             & d_{i,j,\ell} \ge u_{j,\ell}-u_{i,\ell} &,& \{i,j\}\in E \;,\; \ell\in[k]\\[4pt]
             & \sum_{\ell=1}^k u_{i,\ell} = 1 &,& i\in V\\[4pt]
             & u_{i,i}=1 &,& i\in[k] \\
             & u_{i,\ell}\ge 0 &,& i\in V \;,\;\ell\in[k] \\
\end{array}
\]
\caption{The LP relaxation of \MC\ introduced in \cite{CKR00}.}\label{F-LP-relax}
\end{figure}

\section{A Integral Formula for the Density of IT Schemes}\label{A-IT-integral}

We now provide a proof of Lemma~\ref{L-IT-int-formula} which we repeat for convenience.

\begin{replemma}{L-IT-int-formula}
    For every $k\ge 3$ we have:
    \[ d_k^{\IT(f),1}(u_1,\ldots,u_k)
    \LE f(u_1)\cdot \left( \int_0^1 \prod_{i=2}^k(1-tF(u_i))\,dt - \frac{1}{k}\prod_{i=2}^k (1-F(u_i)) \right)\;.\]
\end{replemma}

\begin{proof}
    A convenient way of choosing a random permutation on $[k]$ is to choose for each $i\in [k]$ a uniformly distributed random number $X_i$ in $[0,1]$ and use the permutation that sorts these values in increasing order. Condition on the random number $t=X_1$. Terminal~$i$ catches the edge before it can be cut by terminal~$1$, only if $X_i<t$ and the random threshold chosen for terminal~$i$ is at most~$u_i$. The probability of this event is $t F(u_i)$, as the two events are independent. The events for different $i$'s are also independent. Thus, the probability that no terminal captures the edge before it is cut by terminal~$1$ is exactly $\prod_{i=2}^k(1-tF(u_i))$. Integrating over~$t$ we get the integral term in the statement of the lemma.

    The above expression ignores the fact that terminal~$1$ does not cut the edge if it is the last terminal in the permutation. To obtain the correct expression we need to subtract $\frac{1}{k}\prod_{i=2}^k(1-F(u_i))$ which is the probability that terminal~$1$ is the last in the permutation but would otherwise cut the edge.
\end{proof}

\section{Generalized Exponential Clocks}\label{app:GEC}

Introduced by Buchbinder, Naor, and Schwartz \cite{BNS18}, the \emph{exponential clocks} rounding scheme works as follows. Sample $k$ values $X_1, \hdots, X_k \ge 0$ from the exponential distribution with scale parameter $1$. Then, for any simplex point $u \in \Delta_k$, we color the $u$ according to $\argmax_{i \in [k]} \frac{u_i}{X_i}$ (or equivalently $\argmin_{i \in [k]} \frac{X_i}{u_i}$). We let EC be the resulting rounding scheme. The key property of EC proved by Buchbinder et al.~\cite{BNS18} is that the cut density of $EC$ is $d^{EC}_k(u) = 2 - u_1 - u_2$. This very elementary expression for the cut density allowed for it to be seamlessly analyzed in mixtures with other rounding schemes such as single threshold (ST) and independent thresholds (IT) rounding schemes.

A rather natural avenue toward generalizing the exponential clocks rounding scheme would be to replace the use of the exponential distribution with an arbitrary probability distribution over the positive reals. However, without taking advantage of the unique properties of the exponential distribution, the cut density functions become largely impractical. For completeness, we derive here these cut density functions.
In the competing clock algorithm, we choose $k$ i.i.d. random variables $Z_1, \ldots, Z_k$ from some distribution $\mathcal{D}$. For any $u = (u_1, \ldots, u_k)$, we assign $\ell(u) = \argmin_i\frac{Z_i}{u_i}$. 

We would like to compute the cut density produced by this algorithm. Let $f$, $F$ be the pdf and cdf of $\mathcal{D}$ respectively, and $f_t$, $F_t$ the pdf and cdf of $Z/t$ where $Z \sim \mathcal{D}$. Following the analysis in~\cite{BNS18}, let $u = (u_1, \ldots, u_k), v = (u_1 + \eps, u_2 - \eps, u_3, \ldots, u_k)$, and let $A_i$ be the event that $\ell(u) = \ell(v) = i$. We have
\begin{align*}
    \Pr[A_1] & \EQ \Pr\left[\;\frac{Z_1}{u_1} = \min\left\{\frac{Z_1}{u_1}, \frac{Z_2}{u_2}, \ldots, \frac{Z_k}{u_k}\right\} \;,\; \frac{Z_1}{u_1 + \eps} \EQ \min\left\{\frac{Z_1}{u_1 + \eps}, \frac{Z_2}{u_2 - \eps}, \ldots, \frac{Z_k}{u_k}\right\}\;\right] \\
    & \EQ \Pr\left[\;\frac{Z_1}{u_1} = \min\left\{\frac{Z_1}{u_1}, \frac{Z_2}{u_2}, \ldots, \frac{Z_k}{u_k}\right\}\;\right] \\
    & \EQ \int_0^\infty f_{u_1}(x) \cdot \prod_{j \geq 2} (1 - F_{u_j}(x))\dd x \;, \\
    \Pr[A_2] & \EQ \Pr\left[\;\frac{Z_2}{u_2} = \min\left\{\frac{Z_1}{u_1}, \frac{Z_2}{u_2}, \ldots, \frac{Z_k}{u_k}\right\} \;,\; \frac{Z_2}{u_2 - \eps} \EQ \min\left\{\frac{Z_1}{u_1 + \eps}, \frac{Z_2}{u_2 - \eps}, \ldots, \frac{Z_k}{u_k}\right\}\;\right] \\ 
    & \EQ \Pr\left[ \;\frac{Z_2}{u_2 - \eps} = \min\left\{\frac{Z_1}{u_1 + \eps}, \frac{Z_2}{u_2 - \eps}, \ldots, \frac{Z_k}{u_k}\right\}\;\right] \\
    & \EQ \int_0^\infty f_{u_2-\eps}(x) \cdot (1 - F_{u_1 + \eps}(x)) \cdot \prod_{j \geq 3} (1 - F_{u_j}(x))\dd x \;, \\
    \Pr[A_i] & \EQ \Pr\left[\;\frac{Z_i}{u_i} = \min\left\{\frac{Z_1}{u_1}, \frac{Z_2}{u_2}, \ldots, \frac{Z_k}{u_k}\right\} \;,\; \frac{Z_i}{u_i} \EQ \min\left\{\frac{Z_1}{u_1 + \eps}, \frac{Z_2}{u_2 - \eps}, \ldots, \frac{Z_k}{u_k}\right\}\;\right] \\
    & \EQ \Pr\left[\;\frac{Z_i}{u_i} = \min\left\{\frac{Z_1}{u_1 + \eps}, \frac{Z_2}{u_2}, \ldots, \frac{Z_k}{u_k}\right\}\;\right] \\
    & \EQ \int_0^\infty f_{u_i}(x) \cdot (1 - F_{u_1 + \eps}(x)) \cdot \prod_{j \geq 2, j \neq i} (1 - F_{u_j}(x))\dd x \;,
\end{align*}
where $i \geq 3$. Note that if we plug in the pdf and cdf for exponential random variables ($f_t(x) = te^{-tx}, F_t(x) = 1 - e^{-tx}$), we recover the analysis of Lemma 3.1 in~\cite{BNS18}. When $\eps = 0$, we have that $\sum_{i = 1}^k\Pr[A_i] = 1$, so the cut density is exactly
\[
\lim_{\eps \to 0} \frac{1 - \sum_{i = 1}^k\Pr[A_i]}{\eps} \EQ \left.-\frac{\partial}{\partial \eps} \left(\sum_{i = 1}^k \Pr[A_i]\right)\right|_{\eps = 0} \;.
\]
To compute these derivatives, we use the fact that $f_t(x) = tf(tx)$, $1-F_t(x) = \int_x^\infty f_t(y)\dd y = \int_x^\infty tf(ty)\dd y$. It follows that
\begin{align*}
    \left.\frac{\partial}{\partial \eps} (1 - F_{t + \eps}(x))\right|_{\eps = 0} & \EQ \left.\frac{\partial}{\partial \eps}\int_x^\infty (t+\eps) f((t+\eps)y)\dd y\right|_{\eps = 0} \\
    & \EQ \left.\int_{x}^\infty (t+\eps)\cdot y \cdot f'((t + \eps)y) + f((t+\eps)y) \dd y\right|_{\eps = 0} \\
    & \EQ \int_{x}^\infty t \cdot y \cdot f'(ty) + f(ty) \dd y\\
    & \EQ \left.yf(ty)\right|_x^{+\infty} - \int_{x}^\infty f(ty) \dd y + \int_{x}^\infty f(ty) \dd y \\
    & \EQ -x \cdot f(tx) \;,
\end{align*}
and 
\begin{align*}
    \left.\frac{\partial}{\partial \eps} f_{t-\eps}(x)\right|_{\eps = 0} & \EQ \left.\frac{\partial}{\partial \eps} (t - \eps)f((t-\eps)x)\right|_{\eps = 0} \\
    & \EQ \left. -f((t-\eps)x) - (t-\eps)xf'((t-\eps)x)\right|_{\eps = 0} \\
    & \EQ  -f(tx) - txf'(tx) \;.
\end{align*}
So we have, for $i \geq 3$, 
\begin{align*}
    \left.\frac{\partial}{\partial \eps} \Pr[A_i]\right|_{\eps = 0} \EQ \int_0^\infty f_{u_i}(x) \cdot (-x \cdot f(u_1x)) \cdot \prod_{j \geq 2, j \neq i} (1 - F_{u_j}(x))\dd x \;,
\end{align*}
and
\begin{align*}
    \left.\frac{\partial}{\partial \eps} \Pr[A_2]\right|_{\eps = 0} & \EQ \int_0^\infty (-f(u_2x) - u_2xf'(u_2x)) \cdot \prod_{j \neq 2} (1 - F_{u_j}(x))\dd x \\
    & \qquad + \int_0^\infty f_{u_2}(x) \cdot (-x\cdot f(u_1x)) \cdot \prod_{j \geq 3} (1 - F_{u_j}(x))\dd x\;.
\end{align*}
For the first term in the sum, we have
\begin{align*}
    &\int_0^\infty (-f(u_2x) - u_2xf'(u_2x)) \cdot \prod_{j \neq 2} (1 - F_{u_j}(x))\dd x \\
     =\,\, &  -\int_0^\infty f(u_2x)\cdot \prod_{j \neq 2} (1 - F_{u_j}(x))\dd x - \int_0^\infty u_2xf'(u_2x) \cdot \prod_{j \neq 2} (1 - F_{u_j}(x))\dd x \\
     =\,\, &  -\int_0^\infty f(u_2x)\cdot \prod_{j \neq 2} (1 - F_{u_j}(x))\dd x - \left(\left.x\prod_{j \neq 2} (1 - F_{u_j}(x))f(u_2x)\right|_0^\infty - \int_0^\infty f(u_2x) \dd(x\prod_{j \neq 2} (1 - F_{u_j}(x))) \right)\\     
     =\,\, &  -\int_0^\infty f(u_2x)\cdot \prod_{j \neq 2} (1 - F_{u_j}(x))\dd x + \int_0^\infty f(u_2x) \dd(x\prod_{j \neq 2} (1 - F_{u_j}(x))) \\
     =\,\, & \sum_{j \neq 2}\int_{0}^\infty f(u_2x) (-xf_{u_j}(x))\prod_{\ell \neq 2,j}(1 - F_{u_\ell}(x))\dd x \;.
\end{align*}
In summary, we have that the cut density is equal to
\begin{align*}
    \left.-\frac{\partial}{\partial \eps} \left(\sum_{i = 1}^k \Pr[A_i]\right)\right|_{\eps = 0} & \EQ \frac{1}{u_1}\sum_{i \neq 1}\int_0^\infty x f_{u_1}(x) f_{u_i}(x)\prod_{j \neq 1, i}(1 - F_{u_j}(x)) \dd x \\
    & \qquad + \frac{1}{u_2}\sum_{i \neq 2}\int_0^\infty x f_{u_2}(x) f_{u_i}(x)\prod_{j \neq 2, i}(1 - F_{u_j}(x)) \dd x \;. \\
\end{align*}
As a sanity check, for exponential random variables the above simplifies to $(1 - u_1) + (1 - u_2) = 2 - u_1 - u_2$.

\end{document}